\documentclass[a4paper,aps,prl,twocolumn,10pt,superscriptaddress]{revtex4-1}

\usepackage[utf8]{inputenc}
\usepackage[T1]{fontenc}
\usepackage[sc,osf]{mathpazo}
\usepackage{amsmath, amsthm, amssymb}
\usepackage{graphicx}
\usepackage{dcolumn}
\usepackage{dsfont}
\usepackage{bm}
\usepackage{bbm}
\usepackage{hyperref}
\usepackage{tikz}

\usetikzlibrary{calc}
\usetikzlibrary{decorations.pathmorphing}

\newtheorem{defn}{Definition}
\newtheorem{sdefn}{Supplemental Definition}
\newtheorem{thm}[defn]{Theorem}
\newtheorem{sthm}[sdefn]{Supplemental Theorem}
\newtheorem{prop}[defn]{Proposition}
\newtheorem{sprop}[sdefn]{Supplemental Proposition}
\newtheorem{cor}[defn]{Corollary}

\newtheorem{slem}[sdefn]{Supplemental Lemma}

\theoremstyle{definition}

\newtheorem*{rem*}{Remark}
\newtheorem{innercustomthm}{Theorem}
\newtheorem{innercustomprop}{Proposition}
\newenvironment{customthm}[1]
  {\renewcommand\theinnercustomthm{#1}\innercustomthm}
  {\endinnercustomthm}
\newenvironment{customprop}[1]
  {\renewcommand\theinnercustomprop{#1}\innercustomprop}
  {\endinnercustomprop}
\newtheorem{ex}[defn]{Example}

\newcommand{\R}{\mathbb{R}}
\newcommand{\N}{\mathbb{N}}
\newcommand{\C}{\mathbb{C}}

\newcommand{\ket}[1]{| #1 \rangle}
\newcommand{\bra}[1]{\langle #1 |}
\newcommand{\braket}[2]{\left\langle #1 \mid #2 \right\rangle}
\newcommand{\ketbra}[2]{\left|#1\right\rangle\!\!\left\langle #2\right|}

\newcommand{\tr}{\mathrm{Tr}}

\newcommand{\etal}{{\it{et al.}}}
\newcommand{\proj}[1]{\ensuremath{|#1\rangle \langle #1|}}

\renewcommand{\rho}{\varrho}

\newcommand{\D}{\mathrm{d}}
\newcommand{\supp}{\mathrm{supp}}

\newcommand{\spec}{\mathrm{spec}}

\newcommand{\Hi}{\mathcal{H}}
\newcommand{\St}{\mathbb{S}}
\newcommand{\hi}{\Hi}

\newcommand{\End}[1]{\mathrm{End}\left(#1\right)}

\newcommand{\eps}{\varepsilon}

\newcommand{\relent}[4]{D_{#1}\left( #3 #2\| #4\right)}

\newcommand{\hmin}{H_{\min}}

{\begin{framed}\begin{small}}
{\end{small}\end{framed}}

\pgfmathsetmacro{\rad}{.7}

\DeclareGraphicsExtensions{.pdf,.png,.jpg}

\makeindex

\begin{document}

\title{Catalytic Decoupling of Quantum Information}
\date{\today}

\author{Christian Majenz}
\email{majenz@math.ku.dk}
\affiliation{Department of Mathematical Sciences, University of Copenhagen, Universitetsparken 5, DK-2100 Copenhagen Ø.}
\author{Mario Berta}
\affiliation{Institute for Quantum Information and Matter, California Institute of Technology, Pasadena, CA 91125, USA.}
\author{Fr\'ed\'eric Dupuis}
\affiliation{Faculty of Informatics, Masaryk University, Brno, Czech Republic.}
\author{Renato Renner}
\affiliation{Institute for Theoretical Physics, ETH Zurich, 8093 Z\"urich, Switzerland.}
\author{Matthias Christandl}
\affiliation{Department of Mathematical Sciences, University of Copenhagen, Universitetsparken 5, DK-2100 Copenhagen Ø.}

\begin{abstract}
The decoupling technique is a fundamental tool in quantum information theory with applications ranging from quantum thermodynamics to quantum many body physics to the study of black hole radiation. In this work we introduce the notion of catalytic decoupling, that is, decoupling in the presence of an uncorrelated ancilla system. This removes a restriction on the standard notion of decoupling, which becomes important for structureless resources, and yields a tight characterization in terms of the max-mutual information. Catalytic decoupling naturally unifies various tasks like the erasure of correlations and quantum state merging, and leads to a resource theory of decoupling.
\end{abstract}

\maketitle


\paragraph{Introduction.}

Erasing correlations between quantum systems via local operations, decoupling, is a task that was first studied in the context of quantum information theory~\cite{Horodecki05} (see~\cite{Hayden11} for an introductory tutorial). In particular, decoupling has been crucial for understanding how to distribute quantum information between different parties~\cite{horodecki2007quantum,abeyesinghe2009mother,Luo09,yard2009optimal,Devetak08} and for understanding how to send quantum information over noisy quantum channels~\cite{HHWY08,fred-these,Bennett06,berta2011quantum}. In that context, the idea of decoupling has also been made use of in quantum cryptography~\cite{Berta14}. The concept is, however, also very useful in physics (as, e.g., outlined in~\cite{dupuis2014one}). Applications range from quantum thermodynamics~\cite{RARDV10,Aberg13,diploma-hutter,chaves2012entropic,Brandao12}, to the study of black hole radiation~\cite{HayPre07,braunstein-zyczkowski,braunstein-pati}, and solid state physics~\cite{brandao2011faithful}.


\paragraph{Standard decoupling.}

The basic idea behind decoupling is the following: if a mixed bipartite quantum state $\rho_{AE}$ is only weakly correlated, then it should suffice to erase a small part of $A$ to approximately decouple $A$ from $E$, i.e., to get an approximate product state (see Figure \ref{fig:schematic}). More precisely, we say that a bipartite quantum state $\rho_{AE}$ is $\eps$-decoupled by the partial trace map $T_{A\to A_1}(\cdot)=\tr_{A_2}[\cdot]$ with $A=A_1A_2$ if there exists an unitary operation $U_A$ such that,
\begin{align}\label{eq:decoupling_original}
\min_{\omega_{A_1}\otimes\omega_E}P\big(\mathcal{T}_{A\to A_1}(U_A\rho_{AE}U_A^\dagger),\omega_{A_1}\otimes\omega_E\big)\le \eps,
\end{align}
where the minimum is over all product quantum states $\omega_{A_1}\otimes\omega_E$, and $P(\beta,\gamma):=\left(1-\|\sqrt{\beta}\sqrt{\gamma}\|_1^2\right)^{1/2}$ denotes the purified distance~\cite{mybook}. The $A_1$-system is called the decoupled system and the $A_2$-system the remainder system. Now, the fundamental question that we want to discuss is how large we have to choose the remainder system $A_2$ in order to achieve $\eps$-decoupling. We denote the minimal remainder system size, i.e., the logarithm of the minimal remainder system dimension, for $\eps$-decoupling $A$ from $E$ in a state $\rho_{AE}$ by $R^\eps(A;E)_\rho$.


\paragraph{Converse.}

We first show quite naturally that $R^\eps(A;E)_\rho$ has to be at least of the size of the smooth max-mutual information $I_{\max}^{\eps}(E:A)_{\rho}$ present in the initial state $\rho_{AE}$. This measure is defined as~\cite{berta2011quantum},
\begin{align}
I_{\max}^{\eps}(E;A)_{\rho}&:=\min_{\bar{\rho}}I_{\max}(E;A)_{\bar{\rho}}\quad\mathrm{with}\label{eq:Imax1}\\
I_{\max}(E;A)_{\bar{\rho}}&:=\min_{\sigma_A}\min\left\{\lambda\in\mathbb{R}\middle|2^\lambda\cdot\sigma_A\otimes\bar{\rho}_E\geq\bar{\rho}_{AE}\right\},\label{eq:Imax2}
\end{align}
where the minimum in~\eqref{eq:Imax1} is over all bipartite quantum states with $P(\rho_{AE},\bar{\rho}_{AE})\leq\eps$~\footnote{More precisely, this minimum is taken over sub-normalized states, see supplemental material.}, and the minimum in~\eqref{eq:Imax2} is over all quantum states $\sigma_A$. We note that the definition of the smooth max-information is a priori not symmetric in $A:E$. However, we have~\cite{ciganovic2014smooth},
\begin{align}
I_{\max}^{\eps}(E;A)_{\rho}=_\varepsilon I_{\max}^{\eps}(A;E)_{\rho},
\end{align}
where $=_\varepsilon$ stands for equality up to terms $\mathcal{O}(\log(1/\eps))$. For the converse we exploit that the smooth max-mutual information is invariant under local unitary operations and that it has the so-called non-locking property (see~\cite{divincenzo04} about information locking). That is, just like the quantum mutual information it fulfills the inequality~\cite[Lemma B.12]{berta2011quantum},
\begin{align}\label{eq:lemma_max}
I^{\eps}_{\max}(E;A_1A_2)_\rho\le I^{\eps}_{\max}(E;A_1)_\rho+2\log|A_2|,
\end{align}
where $|A_2|$ denotes the dimension of $A_2$. Since the final state is a product state, its smooth max-mutual information $I^{\eps}_{\max}(E;A_1)_{\omega\otimes\omega}$ becomes zero. This means that in order to erase the initial correlations $I^{\eps}_{\max}(E;A)_\rho$ we need at least a remainder system of size~\footnote{For the definition~\eqref{eq:Imax2}, it is convention in the literature to write $E;A$ (and not $A;E$).},
\begin{align}\label{eq:converse_Imax}
R^\eps(A;E)_\rho\ge \frac 1 2 I_{\max}^{\eps}(E;A)_\rho.
\end{align}


\paragraph{Previous works.}

Most of the aforementioned decoupling references only give good achievability bounds for states of the form $\rho_{A^nE^n}=\rho_{AE}^{\otimes n}$ in the asymptotic limit $n\to\infty$. Whereas this setting is relevant in quantum Shannon theory, it is often a severe restriction for applications in physics. For typical physical situations (e.g., in thermodynamics), there is usually not even a natural decomposition of a large system in $n$ subsystems. A notable exception concerning achievability results is reference~\cite{dupuis2014one}, where the authors show that
\begin{align}\label{eq:achiev_Hmin}
R^{\eps}(A;E)_\rho\leq_\varepsilon \frac{1}{2}\Big(H^{\eps'}_{\max}(A)_\rho-H^{\eps'}_{\min}(A|E)_\rho\Big)\,\mathrm{with}\,\eps'=\frac{\eps}{5},
\end{align}
where $\leq_\varepsilon$ means up to terms $\mathcal{O}(\log(1/\eps))$. (We give a proof of this particular statement in the supplemental material). Here, $H^\eps_{\max}$ and $H^{\eps}_{\min}$ denote the smooth conditional max- and min-entropy whose exact definitions can be found in the supplemental material (or see the textbook~\cite{mybook}). In fact, the results from~\cite{dupuis2014one} also show that not only decoupling in the sense of~\eqref{eq:decoupling_original} is achieved, but moreover that the decoupled system is also randomized. That is, there exists a quantum state $\omega_E$ and a unitary operation $U_A$ such that the decoupled system is left in the fully mixed state:
\begin{align}
P\left(\mathcal{T}_{A\to A_1}(U_A\rho_{AE}U_A^\dagger),\frac{1_{A_1}}{|A_1|}\otimes\omega_E\right)\le \eps.
\end{align}
However, it turns out that there can be an arbitrary big gap between the converse~\eqref{eq:converse_Imax} and the achievability result~\eqref{eq:achiev_Hmin}. This is best seen for an example with trivial system $E$. In that case the achievability bound~\eqref{eq:achiev_Hmin} reduces to the difference between the smooth max- and min-entropy and it is known that this can become roughly as big as $\log|A|$ (we provide an explicit example in the supplemental material). In order to achieve the converse from~\eqref{eq:converse_Imax} we propose in the following a generalized notion of decoupling.


\paragraph{Catalytic decoupling.}

A natural question to ask at this point is if decoupling can be achieved more efficiently in the presence of an already uncorrelated ancilla system (see Figure \ref{fig:schematic}). Formally, we say that $\eps$-decoupling can be achieved catalytically for a bipartite quantum state $\rho_{AE}$ if there exists an ancilla state $\rho_{A'}$ and a decomposition $AA'\cong A_1A_2$ such that
\begin{align}\label{eq:catalytic_decoupling}
\min_{\omega_{A_1}\otimes\omega_E}P\big(\rho_{A_1E},\omega_{A_1}\otimes\omega_E\big)\le \eps\quad\mathrm{where}&\\
\rho_{A_1A_2E}=\rho_{AA'E}=\rho_{AE}\otimes \rho_{A'}&.
\end{align}
Again, we call the $A_1$-system the decoupled system and the $A_2$-system the remainder system. The term catalytic means that the share of the initially uncorrelated ancilla system $A'$ that becomes part of the decoupled system $A_1$ stays decoupled (see Figure \ref{fig:schematic}).

Now, we are interested in the minimal size of the remainder system $A_2$ in order to achieve $\eps$-decoupling catalytically. We denote the optimal remainder system size for catalytically decoupling $A$ from $E$ in a state $\rho_{AE}$ by $R^\eps_c(A;E)_\rho$. Clearly, we have $R^\eps_c(A;E)_\rho\le R^\eps(A;E)_\rho$, as we can always choose a trivial ancilla. Moreover, since appending with an ancilla does not increase the smooth max-mutual information (see supplemental material), the same converse as in~\eqref{eq:converse_Imax} still holds.

These concepts can naturally be phrased as a resource theory of decoupling. A quantum system $A$ coupled to the environment $E$ can yield a decoupled system $A_1$ of a certain size through standard decoupling. That is, in the resource theory language of~\cite{devetak2008resource} we have $\langle\rho_{AE}\rangle\ge_\varepsilon (\log|A|-R^\varepsilon(A;E)_\rho)[d]$. Here, $[d]$ denotes a decoupled qbit and $\ge_\varepsilon$ stands for up to error $\varepsilon$ (see also~\cite{datta11}). Now, our novel paradigm makes use of of the possibility that if we already have decoupled qbits at hand, then we might be able to decouple a larger system,
\begin{align}
\langle\rho_{AE}\rangle+n[d]\ge_\varepsilon \left(n+\log|A|-R^\varepsilon_c(A;E)_\rho\right)[d]&\notag\\
\text{for $n$ large enough.}&
\end{align}


\paragraph{Achievability.}

\begin{figure*}
\includegraphics[width=\textwidth]{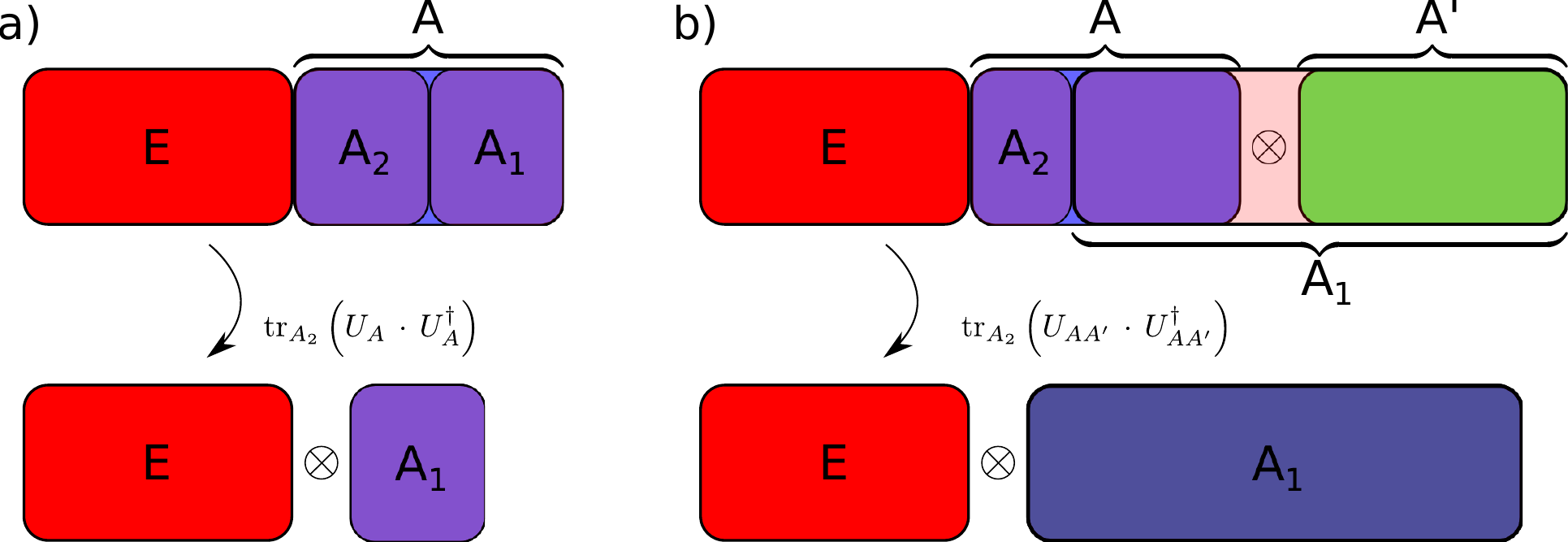}
\caption{Schematic representation of a) standard and b) catalytic decoupling: tracing out a system $A_2$ leaves the remaining state decoupled. While there is no ancilla for standard decoupling as in a), catalytic decoupling as in b) allows to make use of an additional, already decoupled system $A'$. The basic question is how large we have to choose the system $A_2$ such that the remaining system $A_1$ is decoupled from $E$.}\label{fig:schematic}
\end{figure*}

In contrast to standard decoupling as in~\eqref{eq:decoupling_original}, catalytic decoupling can be achieved with a remainder system size that is essentially equal to the smooth max-mutual information.

\begin{thm}[Catalytic decoupling]\label{thm:anc-decoup}
For any bipartite quantum state $\rho_{AE}$ and $0<\delta\leq\eps\leq1$ we have:
\begin{align}\label{eq:decoupcond}
R^\eps_c(A;E)_\rho\lesssim\frac 1 2 I_{\max}^{\eps-\delta}(E;A)_{\rho}
\end{align}
where $\lesssim$ stands for smaller or equal up to terms $\mathcal{O}(\log\log|A|+\log(1/\delta))$. We also have the converse
\begin{align}\label{eq:converse_catalytic}
R^\eps_c(A;E)_\rho\ge \frac 1 2 I_{\max}^{\eps}(E:A)_\rho.
\end{align}
\end{thm}

In fact, we not only show that catalytic decoupling in the sense of~\eqref{eq:catalytic_decoupling} is achieved, but moreover that the decoupled system ends up in the marginal of the original state:
\begin{align}\label{eq:catalytic_stronger}
P\big(\rho_{A_1E},\rho_{A_1}\otimes\omega_E\big)\le \eps\quad\text{for some quantum state $\omega_E$.}
\end{align}
In particular, and in contrast to the standard decoupling results leading to~\eqref{eq:achiev_Hmin}, our catalytic decoupling scheme does not randomize the decoupled system but leaves it invariant (up to the approximation error $\eps$). We can even choose $A_1=AA_1'$ such that the decoupled system contains the marginal of the input state (plus part of the catalyst).

In the supplemental material we give two conceptually different proofs for Theorem~\ref{thm:anc-decoup}. The first proof is based on the standard decoupling techniques from~\cite{berta2011quantum,dupuis2014one} combined with the use of embezzling entangled quantum states~\cite{van2003universal}. For~\eqref{eq:decoupcond} this yields a difference of size at most $\log\log|A|+\mathcal{O}(\log(1/\delta))$~\footnote{The term $\log\log|A|$ can be improved to be logarithmic in the smooth max-information, when accepting a slightly worse leading order term.}. The second proof is based on the convex splitting technique of Anshu {\it et al.}~\cite{anshu2014near}. It allows to upper bound the difference in~\eqref{eq:decoupcond} with the tighter bound
\begin{align}
\frac 1 2 I_{\max}^{\eps-\delta}(E;A)_{\rho}-R^\eps_c(A;E)_\rho\leq\;&\frac 1 2 \Big\{\log\log I_{\max}^{\eps-\delta}(E;A)_{\rho}\Big\}_+&\notag\\
&+\mathcal{O}(\log(1/\delta))\,,
\end{align}
where $\{\cdot\}_+:=\max\{0,\cdot\}$. Moreover, this argument is also constructive and hence leads to an explicit scheme for decoupling. This improves on the standard decoupling bounds which are achieved using the probabilistic technique~\footnote{A partial derandomization can be achieved using (approximate) unitary 2-designs \cite{szehr2013decoupling}.} (as, e.g., the previously best known bound~\eqref{eq:achiev_Hmin} from~\cite{dupuis2014one}).


\paragraph{Discussion.}

The achievability result~\eqref{eq:decoupcond} together with the converse~\eqref{eq:converse_catalytic} establish an operational interpretation of the smooth max-information as twice the minimal size of the remainder system to achieve $\eps$-decoupling. We note that the approximation error as well as the smoothing parameter can be made arbitrarily close in~\eqref{eq:converse_catalytic} and~\eqref{eq:decoupcond} with only a logarithmic penalty. Following the information-theoretic arguments outlined in~\cite{tomamichel2013hierarchy}, we find that for states of the form $\rho_{A^nE^n}=\rho_{AE}^{\otimes n}$ and large $n\to\infty$,
\begin{align}\label{eq:asymptotic_expansion}
&\frac{1}{n}R^\eps_c(A^n;E^n)_{\rho^{\otimes n}}\notag\\
&=\frac{1}{2}\left(I(A:E)_\rho+\sqrt{\frac{V(A:E)_\rho}{n}}\Phi^{-1}(\eps)\right)+\mathcal{O}\left(\frac{\log n}{n}\right),
\end{align}
with the mutual information $I(A:E)_\rho=H(A)_\rho+H(E)_\rho-H(AE)_\rho$ featuring the von Neumann entropy $H(A)_\rho=-\tr(\rho_A\log\rho_A)$, and the mutual information variance $V(A:E)_\rho$, as well as the cumulative normal distribution function $\Phi$ specified in the supplemental material. We note that no such tight (second-order) asymptotic expansion is known for standard decoupling. However, the achievability~\eqref{eq:achiev_Hmin} together with the converse~\eqref{eq:converse_Imax} imply that (using the asymptotic equipartition property from~\cite{mybook}),
\begin{align}
\lim_{n\to\infty}\frac{1}{n}R^\eps(A^n;E^n)_{\rho^{\otimes n}}=\frac{1}{2}I(A:E)_\rho.
\end{align}
Thus, we can conclude that catalytic decoupling and standard decoupling become equivalent in the first order rate asymptotically: the mutual information quantifies the minimal size of the remainder system.


\paragraph{Applications.}

Groisman {\it et al.}~\cite{groisman2005quantum} introduced an operational approach to quantifying the total correlations that are present in a quantum state. In analogy to Landauer's erasure principle~\cite{landauer1961irreversibility}, they characterize the strength of correlations by the amount of randomness that has to be injected locally to decorrelate the state. This randomizing is done by a random-unitary channel on one of the systems (called local unitary randomizing, $A$-LUR in~\cite{groisman2005quantum}):
\begin{align}\label{eq:LUM}
\Lambda(\cdot)=\sum_{i=1}^{N}p_i U_i(\cdot)U_i^\dagger.
\end{align}
We say that that the correlations between $A$ and $E$ in a state $\rho_{AE}$ can be $\eps$-erased by a local mixture of $N$ unitaries on $A$ up to an error $\eps$, if $\Lambda_A$ $\eps$-decouples $A$ from $E$. That is, if there exists a quantum channel $\Lambda_A$ of the form~\eqref{eq:LUM} such that
\begin{align}
\min_{\omega_A\otimes\omega_E}P\left(\Lambda_A(\rho_{AE}),\omega_A\otimes\omega_E\right)\le\eps.
\end{align}
We denote the minimal number of unitaries needed for $\eps$-erasing the correlations between $A$ and $E$ in a state $\rho_{AE}$ by $R^\eps_U(A;E)_\rho$. Groisman {\it et al.}~show that for states of the form $\rho_{A^nE^n}=\rho_{AE}^{\otimes n}$ for large $n\to\infty$:
\begin{align}
\lim_{n\to\infty}\frac{1}{n}R^\eps_U(A^n;E^n)_{\rho^{\otimes n}}=I(A:E)_\rho.
\end{align}
In the following we will see that the task of \emph{catalytic} local erasure of correlation becomes equivalent to catalytic decoupling. We therefore define $R^\eps_{U,c}(A;E)_\rho:=\inf R^\eps_{U}(AA':E)_{\rho\otimes\sigma}$, where the infimum is taken over all ancilla systems.

\begin{prop}[Erasure of correlations]\label{prop:erasure}
For any bipartite quantum state $\rho_{AE}$ we have $\frac{1}{2}R^\eps_{U,c}(A;E)_\rho=R^\eps_c(A;E)_\rho$. Hence, we get (with the notation from Theorem~\ref{thm:anc-decoup}),
\begin{align}
I_{\max}^\eps(E;A)_{\rho}\leq R^\eps_{U,c}(A;E)_\rho\lesssim I_{\max}^{\eps-\delta}(E;A)_{\rho}.
\end{align}
The same asymptotic expansion as in~\eqref{eq:asymptotic_expansion} holds.
\end{prop}

This is the generalization of the results in~\cite{groisman2005quantum} to arbitrary (structureless) states. It gives an alternative operational characterization of the smooth max-mutual information as the the minimal number of unitaries needed for $\eps$-erasing the correlations between $A$ and $E$. The proof of Proposition~\ref{prop:erasure} proceeds as follows. Suppose we have a way of decoupling $A$ from $E$ with remainder system $A_2$, and let $|A_2|=2^k$ for some $k\in \N$. Then, we can think of $A_2$ as $k$ qbits and erase each of them applying a uniform mixture of the Pauli matrices and the identity. This is a mixture of $4^k=2^{2k}$ unitaries. Conversely, suppose we have a uniform mixture of $N=2^{2k}$ unitaries on $A$ that erase the correlations to $E$. We take the fully mixed ancilla state $1_{A'_1A'_2}/|A'_1A'_2|$ with $A'_i\cong\C^{2^k}$. Now, we apply the unitaries controlled on an orthonormal basis of maximally entangled states of $A'_1A'_2$. Then, $A'_1A$ are decoupled from $E$, i.e., we achieved catalytic decoupling with remainder system size $\log |A'_2|=k$.

As a second application we discuss quantum state merging~\cite{Horodecki05} in whose context decoupling was originally introduced~\cite{horodecki2007quantum,abeyesinghe2009mother}. Any catalytic decoupling theorem naturally leads to a quantum state merging protocol. Since the catalytic decoupling theorem is the abstraction of the work on quantum state merging in~\cite{berta2011quantum,anshu2014near}, inserting the bounds from Theorem~\ref{thm:anc-decoup}, we recover the following optimal result for the communication cost of quantum state merging.

\begin{prop}[Coherent quantum state merging]\label{prop:statemerging}
Let $\rho_{ABR}$ be a pure tripartite quantum state shared between Alice, Bob and a Referee. If Alice and Bob have arbitrary entanglement assistance at hand, then Alice can send her system $A$ to Bob up to error $\eps>0$ in purified distance using
\begin{align}\label{eq:statemerging}
q^\eps(A\rangle B)_\rho\lesssim\frac{1}{2}I_{\max}^{\eps/3}(R;\!A)_{\rho}
\end{align}
qbits of quantum communication (with the same notation as in Theorem~\ref{thm:anc-decoup}).
\end{prop}

We note that in the asymptotic limit standard decoupling is sufficient to obtain,
\begin{align}
\lim_{n\to\infty}\frac{1}{n}q^\eps(A^n\rangle B^n)_{\rho^{\otimes n}}=\frac{1}{2}I(R:A)_\rho,
\end{align}
which is also optimal~\cite{abeyesinghe2009mother}. However, for the general setup there is an issue known as entanglement spread~\cite{harrow2009entanglement}, and for the proof of Proposition~\ref{prop:statemerging} we make use of catalytic decoupling and Uhlmann's theorem~\cite{uhlmann85}. In the following we present a proof sketch but defer the full argument to the supplemental material. Setting $\delta=\varepsilon/6$ in Theorem~\ref{thm:anc-decoup} shows that there exists an ancilla state $\rho_{A'}$ and a unitary $U_{AA'\to A_1A_2}$ such that $A_1$ is $\eps/2$ decoupled from $R$ and 
\begin{align}
\log|A_2|\lesssim\;&\frac 1 2 I_{\max}^{\eps/3}(R:A)_{\rho}
\end{align}
Now, Alice and Bob take a pure entangled state $\rho_{A'B'}$ where Alice's part $A'$ is in state $\rho_{A'}$. She applies the unitary $U_{AA'\to A_1A_2}$ and sends $A_2$ to Bob. The decoupling condition and the triangle inequality for the purified distance imply that $P(\rho_{A_1R},\rho_{A_1}\otimes\rho_{R})\le \eps$, so by Uhlmann's theorem there exists a unitary $U_{A_2B\to AB B_1}$ such that
\begin{align}
P(U\rho_{A_1A_2R}U^\dagger,\rho_{A_1B_1}\otimes\rho_{ABR})\le \eps,
\end{align}
where $\rho_{A_1B_1}$ is a purification of $\rho_{A_1}$ and we omitted the subscript of $U$. This implies that Bob has systems $AB$ after applying $U$.

Finally, we also show in the supplemental material that catalytic decoupling directly implies the achievability bound for quantum state redistribution of Anshu {\it et al.}~\cite{anshu2014near} (see~\cite{berta16,Datta14} for alternative bounds).


\paragraph{Extensions.} So far we have analyzed how well the partial trace map $T_{A\to A_1}(\cdot)=\tr_{A_2}[\cdot]$ decouples. However, as originally suggested in~\cite{dupuis2014one}, we can also study quantum channels $T_{A\to B}(\cdot)$ that add noise in an arbitrary way in order to achieve decoupling. To further clarify the important difference between standard decoupling and catalytic decoupling, as well as to correct the faulty~\cite[Corollary 4.2]{dupuis2014one}, we now give a converse for the decoupling behavior of general quantum channels.

\begin{prop}[Correction of Corollary 4.2 from~\cite{dupuis2014one}]\label{prop:corrected_converse}
If for a bipartite quantum state $\rho_{AE}$ and a quantum channel $\mathcal{T}_{A\to B}$,
\begin{align}
\int\D U_AP\left(\mathcal T_{A\to B}(U_A\rho_{AE}U_A^\dagger),T_{A\to B}\left(\frac{1_A}{|A|}\right)\otimes\rho_E\right)\le \eps,
\end{align}
then we have
\begin{align}\label{eq:corrected_converse}
H_{\min}^{\eps'}(A|E)_\rho+H_{\max}^{\eps}(A'|B)_{\tau}\gtrsim0\quad\text{with $\eps'=15\sqrt{\eps}$},
\end{align}
where $\tau_{A'B}=\mathcal{T}_{A\to B}(\phi^+_{A'A})$ is the Choi-Jamio\l kowski state.
\end{prop}

In the supplemental material we prove Proposition~\ref{prop:corrected_converse} starting from~\cite[Theorem 4.1]{dupuis2014one} (from which also the faulty~\cite[Corollary 4.2]{dupuis2014one} was derived). The crucial difference of Proposition~\ref{prop:corrected_converse} to the erroneous version is the assumption that not only decoupling, but decoupling and randomizing is achieved:
\begin{align}
T_{A\to B}\left(\rho_A\right)\otimes\rho_E\quad\mathrm{vs.}\quad T_{A\to B}\left(\frac{1_A}{|A|}\right)\otimes\rho_E.
\end{align}
For example, a product state $\rho_{AE}=\rho_A\otimes\rho_E$ with $\rho_A$ pure has $H_{\min}^{\eps'}(A|E)_\rho\approx0$. It is, however, already perfectly decoupled by the identity map on $A$, which yields $H_{\max}^{\eps}(A|B)_{\tau}\approx -\log |A|$.

In turn, applying the converse bound~\eqref{eq:corrected_converse} to the partial trace map $T_{A\to A_1}(\cdot)=\tr_{A_2}[\cdot]$ shows that the standard decoupling bound~\eqref{eq:achiev_Hmin} in terms of a difference of smooth max- and min-entropy is natural if we ask for decoupling and randomizing. However, if we are not interested in randomizing but only in decoupling, then our main result about catalytic decoupling (Theorem~\ref{thm:anc-decoup}) shows that the smooth max-mutual information is the relevant measure.


\paragraph{Conclusion.}

In this work we introduced the notion of catalytic decoupling. As our main result we established that the optimal remainder system size for decoupling is given by one-half times the smooth max-mutual information. In contrast to standard decoupling results our decoupling scheme is explicit and does not randomize the decoupled system. Moreover, we have shown that catalytic decoupling for general (structureless) states naturally quantifies the resources needed in the erasure of correlation model from~\cite{groisman2005quantum} and for quantum state merging as in~\cite{berta2011quantum}. All of this strengthens the smooth max-mutual information as the one-shot generalization of the quantum mutual information. Finally, given that standard decoupling has already proven useful in various areas of physics (see the references in the introduction), we believe that catalytic decoupling has manifold applications that remain to be explored.


\paragraph{Acknowledgments.}

MC and CM acknowledge financial support from the European Research Council (ERC Grant Agreement no 337603), the Danish Council for Independent Research (Sapere Aude) and the Swiss National Science Foundation (project no PP00P2-150734). MB acknowledges funding provided by the Institute for Quantum Information and Matter, an NSF Physics Frontiers Center (NFS Grant PHY-1125565) with support of the Gordon and Betty Moore Foundation (GBMF-12500028). Additional funding support was provided by the ARO grant for Research on Quantum Algorithms at the IQIM (W911NF-12-1-0521). FD acknowledges the support of the Czech Science Foundation (GA {\v C}R) project no GA16-22211S and FD and RR acknowledge the support of the EU FP7 under grant agreement no 323970 (RAQUEL).


\bibliography{bib}


\onecolumngrid
\newpage

\section{Supplemental Material}

\subsection{Additional notation, definitions and lemmas}

All Hilbert spaces considered here are finite-dimensional. Given a Hilbert space $\hi$ we denote the set of endomorphisms on this Hilbert space by $\End{\hi}$. The set of normalized quantum states on a Hilbert space $\hi$ is denoted by $\St(\hi)=\{\rho\in\End{\hi}|\tr\rho=1, \rho\ge 0\}$, the set of sub-normalized quantum states by $\St_\le(\hi)=\{\rho\in\End{\hi}|\tr\rho\le1, \rho\ge 0\}$. The identity is denoted by $\mathds 1$, and the maximally mixed state by $\tau=\mathds 1/\dim\hi$. The unitary group on this Hilbert space is denoted by $\mathrm{U}(\hi)$. We will make use of two matrix norms, the trace norm and the operator norm, defined by
\begin{align*}
	\|A\|_1&=\tr\sqrt{A^\dagger A}\nonumber\\
	\|A\|_\infty&=\max_{\ket{\phi}\in\hi}\|A\ket{\phi}\|_2,
\end{align*}
for an operator $A\in\End{\hi}$, where $\|\ket{\phi}\|_2=\sqrt{\braket{\phi}{\phi}}$.


\subsubsection{Distance measures}

We need two different metrics on $\St_\le(\hi)$, the trace distance and the purified distance. These are defined as follows.
\begin{sdefn}[Generalized trace distance and purified distance \cite{tomamichel2010duality}]
	For two sub-normalized quantum states $\rho,\sigma\in\St_{\le}(\hi)$, the trace distance is defined as
	\begin{align*}
		\delta(\rho,\sigma)=\frac 1 2 \left(\|\rho-\sigma\|_1+|\tr(\rho-\sigma)|\right).
	\end{align*}
	Their purified distance is defined as
	\begin{align*}
		P(\rho,\sigma)=\sqrt{1-F(\rho,\sigma)^2},\quad\mathrm{where}\quad F(\rho,\sigma)=\|\sqrt{\rho}\sqrt{\sigma}\|_1 +\sqrt{(1-\tr\rho)(1-\tr\sigma)}
	\end{align*}
	is the \emph{generalized fidelity}. We extend these definitions to apply to pairs of probability distributions by considering the corresponding diagonal density matrices. $B_\varepsilon(\rho)$ denotes the purified distance ball of radius $\varepsilon$ around $\rho$, $B^\tr_\varepsilon(\rho)$ the trace distance ball.
\end{sdefn}
Note that the generalized trace distance coincides with the standard definition for normalized states, and the generalized fidelity coincides with the standard fidelity if at least one of the states is normalized.

The two metrics \footnote{It is shown in \cite{tomamichel2010duality} that the generalized trace distance and the generalized purified distance are metrics.} are equivalent and respect the following inequalities.

\begin{slem}[Equivalence of trace distance and purified distance]\label{lem:equiv}
\begin{align*}
	\delta(\rho,\sigma)\le P(\rho,\sigma)\le\sqrt{2\delta(\rho,\sigma)}
\end{align*}
\end{slem}

Forgetting the eigenbases of two states does not increase their trace distance.
\begin{slem}\cite[Box 11.2]{nielsen2010quantum}\label{lem:spectr}
We have
	\begin{align*}
		\delta(\rho,\sigma)\ge\delta(\spec(\rho),\spec(\sigma)),
	\end{align*}
	where $\spec(A)$ denotes the ordered spectrum of a Hermitian matrix $A$.
\end{slem}

The following Lemma is a direct consequence of Supplemental Lemma \ref{lem:equiv} and Hölder's inequality.

\begin{slem}\label{lem:smallU}
	Let $\rho\in\St(\hi)$ be a quantum state and $U\in\mathrm{U}(\hi)$ a unitary. Then, we have
	\begin{align*}
		P(U\rho U^\dagger,\rho)\le \sqrt{2\left\|U-1\right\|_\infty}.
	\end{align*}
\end{slem}

We also need a lemma about low rank approximations of a quantum state.

\begin{slem}\label{lem:lowrk}
	Let $\rho,\rho'\in\St(\hi_A)$ be quantum states. Then, we have
	\begin{align*}
		P(\rho,\rho')\ge P(\rho, \Pi\rho\Pi/\tr(\Pi\rho)),
	\end{align*}
	where $\Pi$ is the projection onto the support of $\rho'$
\end{slem}
\begin{proof}
	Let $\ket{\rho}_{AB}$ be a purification of $\rho$. Then, we have
	\begin{align*}
		F(\rho,\rho')=\max_{\ket{\rho'}}|\braket{\rho}{\rho'}|=\max_{\ket{\rho'}}|\bra{\rho}\Pi\ket{\rho'}|=F(\Pi\rho\Pi,\rho'),
	\end{align*}
	where the maximum is taken over purifications of $\rho'$ on $AB$. But by the Cauchy-Schwarz inequality the normalized vector with maximum inner product with $\Pi\rho\Pi$ is $\Pi\rho\Pi/\tr(\Pi\rho)$, which implies the claimed inequality.
\end{proof}


\subsubsection{Entropies}

In this section we collect additional definitions of entropic quantities that are needed in the proofs given in this supplemental material.

In analogy to the conditional mutual information given in terms of the Shannon entropy, the quantum conditional mutual information is defined in terms of the von Neumann entropy.

\begin{sdefn}[Quantum conditional mutual information]
	The \emph{quantum conditional mutual information} of a tripartite state $\rho_{ABC}\in\St(\hi_A\otimes\hi_B\otimes\hi_C)$ is defined as
	\begin{align*}
		I(A;B|C)_\rho=H(AC)_\rho+H(BC)_\rho-H(ABC)_\rho-H(C)_\rho,
	\end{align*}
	where $H(A)_{\rho}=H(\rho_A)=-\tr(\rho_A\log\rho_A)$ denotes the von Neumann entropy.
\end{sdefn}

In addition to the max-mutual information defined in the main paper we use the following one-shot entropic quantities.

\begin{sdefn}[Max-relative entropy]
	The \emph{max-relative entropy} of a state $\rho\in\St_\le(\hi)$ with respect to a state $\sigma\in\St(\hi)$ is defined as
	\begin{align*}
		D_{\max}(\rho\|\sigma)=\min\left\{\lambda\in\R\Big|2^\lambda\sigma\ge\rho\right\}.
	\end{align*}
\end{sdefn}

\begin{sdefn}[Smooth conditional min- and max-entropy, \cite{renner2005security,tomamichel2010duality}]
	The conditional min-entropy of a positive semidefinite matrix $\rho_{AB}\in\End{\hi_{A}\otimes\hi_B}$ is defined as
	\begin{align*}
		H_{\min}(A|B)_\rho&=\max_{\sigma}\max\left\{\lambda\Big|2^{-\lambda}1_A\otimes\sigma_B\ge \rho_{AB}\right\}\\
		&=\max_\sigma\left(-D_{\max}(\rho_{AB}\big\|\mathds 1_A\otimes\sigma_B\right),
	\end{align*}
	where the maximum is taken over all normalized quantum states. The conditional max-entropy is defined as the dual of the conditional min-entropy in the sense that
	\begin{align*}
	H_{\max}(A|B)_\rho=-H_{\min}(A|C)_\rho,
	\end{align*}
	where $\rho_{ABC}$ is a purification of $\rho_{AB}$. The smooth conditional min- and max-entropies are defined by maximizing and minimizing over a ball of sub-normalized states $\tilde\rho_{AB}$, respectively,
	\begin{align*}
		H_{\min}^\varepsilon(A|B)_\rho=\max_{\tilde{\rho}\in B_\varepsilon(\rho)}H_{\min}(A|B)_{\tilde{\rho}},\quad\mathrm{and}\quad H_{\max}^\varepsilon(A|B)_\rho=\min_{\tilde{\rho}\in B_\varepsilon(\rho)}H_{\max}(A|B)_{\tilde{\rho}}.
	\end{align*}
\end{sdefn}

The conditional max-entropy can be expressed in terms of the fidelity.

\begin{slem}\cite[Theorem 3]{konig2009operational}
We have
\begin{align*}
H_{\max}(A|B)_\rho=\max_{\sigma\in\St(\hi_B)}2\log F(\rho_{AB},\mathds{1}_A\otimes\sigma_B).
\end{align*}
\end{slem}

The unconditional min- and max-entropy are defined as their conditional counterparts with a trivial conditioning system.

\begin{slem}\cite{konig2009operational}\label{lem:uncond}
	The min and max-entropy are given by
	\begin{align*}
		H_{\min}(\rho)=-\log\|\rho\|_\infty\quad\mathrm{and}\quad H_{\max}(\rho)=2\log\tr\sqrt{\rho}.
	\end{align*} 
\end{slem}

The min-entropy does not decrease under projections\footnote{This is counterintuitive and due to the fact that we only project but do not renormalize the state.}.

\begin{slem}\label{lem:hminproj}
	Let $\rho_{AB}$ be a bipartite quantum state and $\Pi_A$ a projection on $\hi_A$. Then, we have
	\begin{align*}
		H_{\min}(A|B)_\rho\le H_{\min}(A|B)_{\Pi\rho\Pi}.
	\end{align*}
\end{slem}

\begin{proof}
	Let $\sigma_B$ be a quantum state such that $2^{-H_{\min}(A|B)_\rho}1_A\otimes\sigma_B\ge\rho_{AB}$. Applying $\Pi$ on both sides yields
	\begin{align*}
			2^{-H_{\min}(A|B)_\rho}1_A\otimes\sigma_B\ge 2^{-H_{\min}(A|B)_\rho}\Pi_A\otimes\sigma_B\ge\Pi_A\rho_{AB}\Pi_A.
		\end{align*}
		This is a valid point in the maximization defining $H_{\min}(A|B)_{\Pi\rho\Pi}$, implying the result.
\end{proof}

The fact that min- and max-entropy are invariant under local isometries~\cite[Lemma 13]{tomamichel2010duality}, implies that the max-mutual information has the same property.

\begin{slem}
	For a bipartite quantum state $\rho_{AB}$ and isometries $V_{A\to A'},\ W_{B\to B'}$,
	\begin{align*}
		I_{\max}^{\varepsilon}(A;B)_\rho=I_{\max}^{\varepsilon}(A';B')_{\tilde{\rho}},
	\end{align*}
	where $\tilde{\rho}_{A'B'}=V\otimes W\rho_{AB} V^\dagger\otimes W^\dagger$
\end{slem}

\begin{proof}
	Suppose first that $\rho_A$ is invertible. Then, it follows directly from the definitions of the max-mutual information and the conditional min-entropy, that
	\begin{align*}
			I_{\max}(A;B)_\rho=-H_{\min}(A|B)_{\rho_{B|A}},
	\end{align*}
	where $\rho_{B|A}=\rho_{A}^{-1/2}\rho_{AB}\rho_{A}^{-1/2}$. This implies, together with the case $\varepsilon=0$ of~\cite[Lemma 13]{tomamichel2010duality} that the non-smooth max information is invariant under isometries.
	
	Now, we treat the smooth case. There exists a state $\bar{\rho}_{AB}$ such that $I_{\max}^{\varepsilon}(A;B)_\rho=I_{\max}(A;B)_{\bar\rho}$, i.e.
	\begin{align*}
		I_{\max}^{\varepsilon}(A;B)_\rho=I_{\max}(A;B)_{\bar\rho}=I_{\max}(A';B')_{\tilde{\bar\rho}}\ge I_{\max}^\varepsilon(A';B')_{\tilde{\rho}},
	\end{align*}
	where $\tilde{\bar\rho}_{A'B'}=V\otimes W\bar\rho_{AB} V^\dagger\otimes W^\dagger$.
	The other inequality is proven in a way similar to the one in \cite[Lemma 13]{tomamichel2010duality}. Let $\tau_{A'B'}\in B_\varepsilon(\tilde\rho),\  \eta_{B'}$ be quantum state such that $2^\lambda\tau_{A'}\otimes\eta_{B'}\ge\tau_{A'B'}$, where $\lambda=I_{\max}^\varepsilon(A';B')_{\tilde{\rho}}$. Let $\Pi_V, \Pi_W$ be the projections onto the ranges of $V$ and $W$. It follows that $2^\lambda\bar\tau_{A'}\otimes\bar\eta_{B'}\ge\bar\tau_{A'B'}$, where $\bar\tau_{A'B'}=\Pi_V\otimes\Pi_W\tau_{A'B'}\Pi_V\otimes\Pi_W$ and $\bar\eta_{B'}=\Pi_W\eta_{B'}\Pi_W$. It follows from the fact that the purified distance contracts under projections that $\bar\tau_{A'B'}\in B_\varepsilon(\tilde{\rho})$ and therefore
	\begin{align*}
			I_{\max}^{\varepsilon}(A';B')_{\tilde{\rho}}\ge\min\left\{\lambda\in\R\big|2^\lambda\tilde\tau_{A}\otimes\tilde\eta_{B}\ge\tilde\tau_{AB}\right\}\ge I_{\max}^{\varepsilon}(A;B)_\rho,
		\end{align*}
	where $\tilde{\tau}_{AB}= V^\dagger\otimes W^\dagger\bar\tau_{A'B'} V\otimes W$ and $\tilde\eta_{B}=W^\dagger\bar\eta_{B}W$. The observation that the minimum can be replaced by an infimum over invertible states in the definition of the smooth max-mutual information finishes the proof.
\end{proof}

Tensoring a local ancilla does not change the max-mutual information.

\begin{slem}
	Let $\rho_{AB}$, $\sigma_C$ be quantum states. The smooth max-mutual information is invariant under adding local ancillas,
	\begin{align*}
		I_{\max}^{\varepsilon}(A;B)_\rho=I_{\max}^{\varepsilon}(A;BC)_{\rho\otimes\sigma}
	\end{align*}
\end{slem}

\begin{proof}
	According to~\cite[Lemma B.17]{berta2011quantum} the max-mutual information decreases under local CPTP maps. But both adding and removing an ancilla is such a map, which implies the claimed invariance.
\end{proof}

There are several ways to define the max-mutual information \cite{ciganovic2014smooth}, one of the alternative definitions will be useful for catalytic decoupling.

\begin{sdefn}\label{defn:alt-max-mut}
	An alternative max-mutual information of a quantum state $\rho_{AB}$ is defined by
	\begin{align*}
		I_{\max}(A:B)_{\rho,\rho}=D_{\max}(\rho\|\rho_A\otimes\rho_B).
	\end{align*}
	The smooth version $I_{\max}^\varepsilon(A:B)_{\rho,\rho}$ is defined analogously to $I^\varepsilon_{\max}(A:B)_{\rho}$,
	\begin{align*}
		I_{\max}^\varepsilon(A:B)_{\rho,\rho}=\min_{\tilde{\rho}\in B_\varepsilon(\rho)}I_{\max}(A:B)_{\tilde{\rho},\tilde{\rho}}
	\end{align*}
\end{sdefn}

This alternative definition has some disadvantages, in particular the non-smooth version is not bounded from above for a fixed Hilbert space dimension. The  two different smooth max-mutual informations, however, are quite similar, in particular they can be approximated up to a dimension independent error.

\begin{slem}[\cite{ciganovic2014smooth}, Theorem 3]\label{lem:ciganovic}
	For a bipartite quantum state $\rho_{AB}$,
	\begin{equation}
		I_{\max}^{\varepsilon+2\sqrt{\varepsilon}+\varepsilon'}(A:B)_\rho\lesssim I_{\max}^{\varepsilon+2\sqrt{\varepsilon}+\varepsilon'}(A:B)_{\rho,\rho}\lesssim I_{\max}^{\varepsilon'}(A:B)_\rho,
	\end{equation}
	where the notation $\lesssim$ hides errors of order $\log(1/\varepsilon)$ as in the main text.
\end{slem}

As an auxiliary quantity we also need the unconditional R\'enyi entropy of order $0$.

\begin{sdefn}
	For a quantum state $\rho_A\in\St(\hi_A)$ the R\'enyi entropy of order 0 is defined by
	\begin{align*}
		H_0(A)_\rho=\log\mathrm{rk}(\rho_A),
	\end{align*}
	where $\mathrm{rk}(X)$ denotes the rank of a matrix $X$. Like in the case of the max-entropy, the smoothed version is defined by minimizing over an epsilon ball,
	\begin{align*}
		H_{0}^\varepsilon(A)_\rho=\min_{\tilde{\rho}\in B_\varepsilon(\rho)}H_{0}(A)_{\tilde{\rho}}.
	\end{align*}
\end{sdefn}

The smoothed $0$-entropy is almost equal to the smoothed max-entropy.

\begin{slem}\cite[Lemma 4.3]{renner2004smooth}\label{lem:H0Hmax}
We have
	\begin{align*}
		H_{\max}^{2\varepsilon}(\rho)\le H_0^{2\varepsilon}(A)_\rho\le H_{\max}^\varepsilon(\rho)+2\log(1/\varepsilon).
	\end{align*}
\end{slem}


\subsection{Examples and proofs}

Here we give proofs for the theorems given and claims made in the main paper, and explicit examples.


\subsubsection{Catalytic decoupling}

Here we present two proofs Theorem 1 in the main text, the achievability of catalytic decoupling.

The following is the key lemma of \cite{anshu2014near} and called convex split lemma by the authors.

\begin{slem}\cite[Lemma 3.1]{anshu2014near}\label{lem:convsplit}
	Let $\rho\in\St(\hi_{A}\otimes\hi_E)$ and $\sigma\in\St(\hi_E)$ be quantum states, $k=D_{\text{max}}(\rho_{AE}\|\rho_A\otimes\sigma_E)$ and $0<\delta<\frac{1}{6}$. Define
	\begin{align*}
		n=\begin{cases}
			1 & k\le 3\delta\\ \left\lceil\frac{8\cdot 2^{k}\log\left(\frac{k}{\delta}\right)}{\delta^3}\right\rceil&\mathrm{else}
		\end{cases}.
	\end{align*}
	For the state
	\begin{align}\label{eq:taudef}
		\tau_{AE_1...E_n}=\frac{1}{n}\sum_{j=1}^n\rho_{AE_j}\otimes\left(\sigma^{\otimes (n-1)}\right)_{E_{j^c}}
	\end{align}
	$E$ is decoupled from $A$ in the following sense:
	\begin{align*}
		I(A;E_1...E_n)_\tau\le3\delta\quad\text{as well as}\quad P\left(\tau_A\otimes\tau_{E_1...E_n},\tau_{AE_1...E_n}\right)\le\sqrt{6\delta},
	\end{align*}
	where $E_{j^c}$ denotes $\{E_i\}_{i\neq j}$.
\end{slem}

\begin{customthm}{1}[Catalytic decoupling]\label{thm:anc-decoup-supp}
	Let $\hat\rho_{AE}\in\St(\hi_A\otimes \hi_E)$ be a quantum state. Then, for any $0<\delta\le \varepsilon$ catalytic decoupling with error $\varepsilon$ can be achieved with remainder system size
	\begin{align*}
		\log|A_2|\le\frac 1 2 \left(I_{\max}^{\varepsilon-\delta}(E;A)_{\hat\rho}+\left\{\log\log I_{\max}^{\varepsilon-\delta}(E;A)_{\hat\rho}\right\}_+\right)+\mathcal{O}(\log\frac{1}{\delta}),
	\end{align*}
	where we define $\{x\}_+$ to be equal to $x$ if $x\in\R_{\ge0}$ and $0$ otherwise.
\end{customthm}

\begin{proof}
	Let $\gamma=\varepsilon-\delta$. Take $\rho\in B_\gamma(\hat{\rho})$ such that $I_{\max}(E;A)_\rho=I^\gamma_{\max}(E;A)_{\hat\rho}$.
	Let $\sigma_A$ be the minimizer in 
	\begin{align*}
	k=I_{\max}(E;A)_\rho=\min_{\sigma_A\in\St(\hi_A)}\relent{\max}{\big}{\rho_{AE}}{\sigma_A\otimes\rho_E}.
	\end{align*}
	If $k\le \frac{\delta^2}{2}$ the state is already decoupled according to Supplemental Lemma \ref{lem:convsplit} and the statement is trivially true, so let us assume $k>\frac{\delta^2}{2}$.
	We want to use Supplemental Lemma \ref{lem:convsplit} so let
	\begin{align*}
		n=\left\lceil\frac{8\cdot 2^{k}\log\left(\frac{k}{\delta'}\right)}{\delta'^3}\right\rceil
	\end{align*}
	with $\delta'=\frac{\delta^2}{6}$, $\hi_{A'}=\hi_A^{\otimes (n-1)}\otimes\hi_{\bar A}$ with $\hi_{\bar A}\cong\C^n$ and define the state $\tilde{\rho}_{A^{(2)}...A^{(n)}\bar A}=\sigma^{\otimes (n-1)}\otimes\tau_{\bar A}$, where $\tau_{\bar{A}}=\mathds{1}_{\bar A}/|\bar A|$ denotes the maximally mixed state on $\hi_{\bar A}$. We can now define a unitary that permutes the $A$-systems conditioned on $\bar A$ and thus creates an extension of the state $\tau$ from Equation \eqref{eq:taudef} when applied to $\rho_{AE}\otimes \tilde\rho_{A'}$,
	\begin{align*}
	U^{(1)}_{AA'}=\sum_{j=1}^n(1j)_{A^{(1)}...A^{(n)}}\otimes\ketbra{j-1}{j-1}_{A'},
	\end{align*}
	where $(1j)_{A^{(1)}...A^{(n)}}$ is the transposition $(1j)\in S_n$ under the representation $S_n\looparrowright \hi_A^{\otimes n}$ of the symmetric group that acts by permuting the tensor factors, and $(11)=1_{S_n}$. Now, we are almost done, as Supplemental Lemma \ref{lem:convsplit} implies that
	\begin{align*}
	P\left(\xi_{EA^{(1)}...A^{(n)}},\xi_E\otimes\xi_{A^{(1)}...A^{(n)}}\right)\le \delta,
	\end{align*}
	where $\xi=U^{(1)}_{AA'}\rho_{AE}\otimes\tilde\rho_{A'} \left(U^{(1)}_{AA'}\right)^\dagger$. The register $\bar A$, however, is still a factor of two larger than the claimed bound for $|A_2|$. We can win this factor of two by using superdense coding, as $\bar A$ is classical. Let us therefore slightly enlarge $\hi_{\bar A}$ such that $\dim(\hi_{\bar A})=m^2$ for $m=\lceil\sqrt n\rceil$. We now rotate the standard basis of $\hi_M$ into a Bell basis
	\begin{align*}
		\ket{\psi_{kl}}=\frac{1}{\sqrt m}\sum_{s=0}^{m-1}e^{\frac{2\pi i k s}{m}}\ket s\otimes\ket{s+l\mod m}
	\end{align*}
	 of $\hi_{\bar A_1}\otimes\hi_{\bar A_2}$, with $\hi_{\bar A_i}\cong\C^m$. That is done by the unitary
	 \begin{align*}
	 	U^{(2)}_{\bar A}: \hi_{\bar A}\to\hi_{\bar A_1}\otimes\hi_{\bar A_2}\quad\mathrm{with}\quad U^{(2)}_{\bar A}=\sum_{k,l=0}^{m-1}\ketbra{\psi_{kl}}{m k+l}.
	 \end{align*}
	 As $\tr_{\bar A_2}\proj{\psi_{kl}}=\tau_{\bar A_1}$ for all $k,l\in\{0,...,m-1\}$, the unitary $V_{AA'\to A_1A_2}=U^{(2)}_{\bar A}U^{(1)}_{AA'}$ and the definitions $\hi_{A_2}=\hi_{\bar A_2}$ and $\hi_{A_1}=\hi_A^{\otimes n}\otimes\hi_{\bar A_1}$ achieve $P\left(\eta_{A_1 E},\eta_{A_1}\otimes\rho_E\right)\le \delta$, with $\eta=V_{AA'\to A_1A_2}\rho\otimes\tilde{\rho}V_{AA'\to A_1A_2}^\dagger$. Using the triangle inequality for the purified distance we finally arrive at
	 \begin{align*}
	 	P\left(\hat\xi_{A_1E},\eta_{A_1}\otimes\rho_E\right)\le \gamma+\delta=\varepsilon,
	 \end{align*}
	 for $\hat\xi=V_{AA'\to A_1A_2}\hat\rho\otimes\tilde\rho V_{AA'\to A_1A_2}^\dagger$.
	 The size of the remainder system is 
	 \begin{align*}
	 	\log|A_2|=\frac 1 2\log n\le\frac 1 2 \big(I_{\max}^\gamma(E;A)_{\hat\rho}+\left\{\log\log I_{\max}^\gamma(E;A)_{\hat\rho}\right\}_+\big)+\mathcal{O}\left(\log\frac{1}{\delta}\right).
	 \end{align*}
\end{proof}

\begin{rem*}
	Using the alternative definition of the max-mutual information, Supplemental Definition \ref{defn:alt-max-mut}, we can prove in the same way that
	\begin{align*}
		P(\hat\xi_{A_1E}, \eta_{A_1}\otimes\rho_E)\le \varepsilon
	\end{align*}
	with $\hat\xi=V_{AA'\to A_1A_2}\hat\rho_{AE}\otimes\rho_A^{\otimes n} V_{AA'\to A_1A_2}^\dagger$ and $\eta=V_{AA'\to A_1A_2}\rho_{AE}\otimes\rho_A^{\otimes n} V_{AA'\to A_1A_2}^\dagger$ in this case, and $n$ defined in the same way as above, just with $k=I_{\max}^\varepsilon(A:E)_{\hat\rho, \hat{\rho}}$. This achieves a stronger notion of decoupling, as a large part of the catalyst can be approximately handed back in the same state,
	\begin{align*}
		\eta_{A_1}=\rho_{A}^{\otimes n}\otimes \tau_{\bar A_1}.
	\end{align*}
	By Supplemental Lemma \ref{lem:ciganovic} this still implies
	\begin{align*}
		\log|A_2|&=\frac 1 2\log n\le\frac 1 2 \big(I_{\max}^{\varepsilon-\delta}(E;A)_{\hat\rho,\hat\rho}+\left\{\log\log
		I_{\max}^{\varepsilon-\delta}(E;A)_{\hat\rho,\hat\rho}\right\}_+\big)+\mathcal{O}\left(\log\frac{1}{\delta}\right)\\
		&\le\frac 1 2 \big(I_{\max}^{\varepsilon-2\delta-2\sqrt\delta}(E;A)_{\hat\rho}+\left\{\log\log
				I_{\max}^{\varepsilon-2\delta-2\sqrt\delta}(E;A)_{\hat\rho}\right\}_+\big)+\mathcal{O}\left(\log\frac{1}{\delta}\right)\\
				&\le\frac 1 2 \big(I_{\max}^{\varepsilon-\delta'}(E;A)_{\hat\rho}+\left\{\log\log
								I_{\max}^{\varepsilon-\delta'}(E;A)_{\hat\rho}\right\}_+\big)+\mathcal{O}\left(\log\frac{1}{\delta'}\right),
	 \end{align*}
	 having defined $\delta'=2\delta+2\sqrt{\delta}$.
\end{rem*}

The second proof is based on the state splitting protocol in \cite{berta2011quantum}. This uses \emph{embezzling states} \cite{van2003universal}.

\begin{sdefn}[Embezzling state \cite{van2003universal}]
	A state $\ket\mu\in\mathcal{H}_{A}\otimes\hi_B$ is called \emph{universal $(d,\delta)$-embezzling state} if for any state $\ket{\psi}\in\hi_{A'}\otimes\hi_{B'}$ with $\dim \hi_A=\dim\hi_B\le d$ there exists an isometry $V_{\psi,X}:\hi_{X}\to\hi_X\otimes\hi_{X'}$, $X=A,B$ such that\footnote{The original concept was defined using the trace distance instead of the purified distance \cite{van2003universal}. We use the purified distance here as it fits our task, the definitions are equivalent up to a square according to Supplemental Lemma \ref{lem:equiv}.}
	\begin{align*}
	P\left(V_{\psi,A}\otimes V_{\psi,B}\ket\mu,\ket\mu\otimes\ket\psi\right)\le\delta.
	\end{align*}
\end{sdefn}

\begin{sprop}\cite{van2003universal}
Universal $(d,\delta)$-embezzling states exist for all $d$ and $\varepsilon$.
\end{sprop}

The proof also uses the one-shot version of standard decoupling.

\begin{slem}\cite[Theorem 3.1]{berta2011quantum}, \cite[Table 2]{dupuis2014one}\label{cor:bcr-decoup}
	Let $\rho_{AE}\in\St(\hi_A\otimes \hi_E)$ be a quantum state. Then, there exists a decomposition $\hi_A\cong\hi_{A_1}\otimes\hi_{A_2}$ with
	\begin{align*}
		\log(|A_2|)\le \frac{1}{2}\left(\log(|A|)-H_{\min}(A|E)_\rho\right)+2\log\frac 1 \varepsilon+1.
	\end{align*}
	 such that
	 \begin{align*}
	 	P\left(\rho_{A_1E},\frac{1_{A_1}}{|A_1|}\otimes\rho_E\right)\le\varepsilon.
	 \end{align*}
\end{slem}

The difference between the bound given here and the bound from \cite[Theorem 3.1]{berta2011quantum} stems from the fact that we define decoupling using the purified distance.

We include the following alternative proof to show how catalytic decoupling unifies different techniques from one-shot quantum communication. As a first step we prove the following non-smooth theorem.

\begin{sthm}[Non-smooth catalytic decoupling from standard decoupling and embezzling states]\label{thm:nonsmooth-anc-decoup-bcr}
Let $\rho_{AE}\in\St(\hi_A\otimes \hi_E)$ be a quantum state. Then, $\varepsilon$-catalytic decoupling can be achieved with remainder system size
	\begin{align*}
		\!\log|A_2|\!\le\frac{1}{2}I_{\max}(A;E)_{\rho}+\log H_0(A)_{{\rho}}+\mathcal{O}\left(\log\left(\frac 1 \varepsilon\right)\right).
	\end{align*}
	In addition, if we allow for the use of isometries instead of unitaries, the ancilla systems final state is $\varepsilon$ close to its initial state.
\end{sthm}

\begin{proof}
	For notational convenience let $\ket{\rho}_{AER}$ be a purification of $\rho$. 
	Also in slight abuse of notation we replace $\hi_A$ by $\mathrm{supp}(\rho_A)$ so that $|A|\le2^{ H_0(A)_{{\rho}}} $. The idea is to decompose the Hilbert space $\hi_A$ into a direct sum of subspaces where the spectrum of $\rho_A$ is almost flat. Let $Q=\left\lceil\log|A|+2\log\left(\frac 1 \varepsilon\right)-1\right\rceil$ and define the projectors $P_i,\ i=0,...,Q+1$ such that $P_{Q+1}$ projects onto the eigenvectors of $\rho_A$ with eigenvalues in $\left[0,2^{-(Q+1)}\right]$ and $P_i$ projects onto the eigenvectors of $\rho_A$ with eigenvalues in $\left[2^{-(i+1)},2^{-i}\right]$ for $i=0,...,Q$. We can now write the approximate state $\ket{\bar\rho}_{AER}=\frac{1}{\sqrt \alpha}(\mathds 1_A-P_{Q+1})\ket{\rho}_{AER}, \alpha=\tr(\mathds 1_A-P_{Q+1})\rho$ as a superposition of states with almost flat marginal spectra on $A$,
	\begin{align*}
		\ket{\bar\rho}=\sum\sqrt{p_i}\ket{\rho^{(i)}},
	\end{align*}
	 with $p_i=\tr\bar\rho P_i$ and  $\ket{\rho^{(i)}}=\frac{1}{\sqrt{p_i}}P_i\ket{\rho}$. This decomposition corresponds to the direct sum decomposition 
	 \begin{align*}
	 	\hi_A\cong\bigoplus_{i=0}^{Q+1}\hi_{A^{(i)}},
	 \end{align*}
	 where $\hi_{A^{(i)}}=\mathrm{supp}(P_i)$. Note that we have $P(\rho,\bar\rho)=\sqrt{1-\alpha}$ and $$1-\alpha\le |A|2^{-\left(\log|A|+2\log\left(\frac 1 \varepsilon\right)\right)}=\varepsilon^2,$$ i.e. $P(\rho,\bar\rho)\le\varepsilon$. Now, we have a family of states, $\{\rho^{(i)}_{A^{(i)}E}\}$ to each of which we apply Supplemental Lemma \ref{cor:bcr-decoup}. This yields decompositions $\hi_{A^{(i)}}\cong\hi_{A_1^{(i)}}\otimes\hi_{A_2^{(i)}}$ such that
	 \begin{align}\label{eq:flat-decoup}
	 	P\left(\rho^{(i)}_{A^{(i)}_1E},\tau_{A^{(i)}_1}\otimes \rho^{(i)}_E\right)\le\varepsilon
	 \end{align}
	 and
	 \begin{align*}
	 	\log(|A^{(i)}_2|)\ge \frac{1}{2}\left(\log(|A^{(i)}|)+H_{\min}(A|E)_{\rho^{(i)}}\right)+2\log\left(\frac 1 \varepsilon\right)+1,
	 \end{align*}
	 where $\tau_A=\frac{\mathds{1}_A}{|A|}$ is the maximally mixed state on a quantum system $A$.
	 
	At this stage of the protocol the situation can be described as follows. Conditioned on $i$, $A_1^{(i)}$ is decoupled from $E$. If $\rho^{(i)}_E\neq\rho^{(j)}_E$ and $\left|A^{(i)}\right|\neq\left|A^{(j)}\right|$, however, there are still correlations left between $A_1$ and $E$. To get rid of this problem, we hide the maximally mixed states of different dimensions in an embezzling state by "un-embezzling" them. Let us therefore first isometrically embed all these states in the same Hilbert space. To do that, define
	 \begin{align*}
	  	d_2=\max_i\left|A^{(i)}_2\right|\quad\text{and}\quad d_1=\max\left(\max_i\left|A^{(i)}_1\right|, \left\lceil\frac{\left|A^{(Q+1)}\right|}{d_2}\right\rceil\right). 
	 \end{align*}
	 Now, let $\hi_{\tilde A_\alpha}\cong\C^{d_\alpha}$ and choose isometries $U^{(\alpha,i)}_{A^{(i)}_\alpha\to\tilde A_\alpha}$ for $\alpha=1,2$, define $U^{(i)}_{A^{(i)}\to\tilde A_1\otimes \tilde A_2}=U^{(1,i)}_{A^{(i)}_1\to\tilde A_1}\otimes U^{(2,i)}_{A^{(i)}_2\to\tilde A_2}$ for $i=1,...,Q$. In addition, choose an isometry $U^{Q+1}_{A^{(Q+1)}\to \tilde A_1\otimes \tilde A_2}$. Let $\ket{\mu}_{A'B'}\in\hi_{A'}\otimes\hi_{B'}$ be a $(d_1,\varepsilon)$-embezzling state, and let $\sigma_{A'}=\tr_{B'}\proj \mu$. Define the isometries $\bar V^{(i)}_{A'\to A'\tilde A_1}$ that would embezzle a state $\tau^{(i)}_{\tilde A_i}=U^{(1,i)}_{A^{(i)}_1\to\tilde A_1}\tau_{A^{(i)}_1}\left(U^{(1,i)}_{A^{(i)}_1\to\tilde A_1}\right)^\dagger$ from $\sigma_{A'}$. Taking some state $\ket 0_{\tilde A_1}\in\hi_{\tilde A_1}$ we can pad these embezzling isometries to unitaries $V^{(i)}_{A'\tilde A_1}$ such that
	 \begin{align}\label{eq:padembezz}
		 P\left(V^{(i)}_{A'\tilde A_1}\sigma_{A'}\otimes\proj 0_{\tilde A_1}\left(V^{(i)}_{A'\tilde A_1}\right)^\dagger,\sigma_{A'}\otimes \tau^{(i)}_{\tilde A_i}\right)\le\varepsilon.
	 \end{align}
	 We can combine the above isometries and unitaries now to un-embezzle the states that are approximately equal to $\tau_{A^{(i)}_1}$ conditioned on $i$. Define $\tilde{A}_3\cong\C^{Q+1}$ and 
	 \begin{align*}
	 	W^I_{A\to\tilde A_1\tilde A_2\tilde A_3}=\sum_i U^{(i)}_{A^{(i)}\to \tilde A_1\otimes \tilde A_2}P_i\otimes\ket i_{\tilde A_3},\quad W^{II}_{A'\tilde A_1\tilde A_3}=\sum_i \left(V^{(i)}_{A'\tilde A_1}\right)^\dagger\otimes\proj i_{\tilde A_3}.
	 \end{align*}
	 The final state of our decoupling protocol is
	 \begin{align*}
	 	\rho^f_{A_1A_2E}=W^{II}W^I\rho_{AE}\otimes\sigma_{A'} \left(W^I\right)^\dagger\left(W^{II}\right)^\dagger,
	 \end{align*}
	 where we omitted the subscripts of the $V$s for compactness and have defined $A_1=A'\tilde{A}_1$ and $A_2=\tilde A_2\tilde A_3$. Let us show that this protocol actually decouples $A_1$ from $E$. We bound, omitting the subscripts of unitaries and isometries,
	 \begin{align*}
	 	P\left(\rho^f_{A_1E},\sigma_{A'}\otimes\proj 0_{\tilde A_1}\otimes\rho_E\right)
	 	&=\sum_ip_iP\Big(\left(V^{(i)}\right)^\dagger U^{(1,i)}\sigma_{A'}\otimes\rho^{(i)}_{A^{(i)}_1E}\left(U^{(1,i)}\right)^\dagger V^{(i)},\sigma_{A'}\otimes\proj 0_{\tilde A_1}\otimes\rho_E\Big)\nonumber\\
	 	&=\sum_ip_iP\Big( U^{(1,i)}\sigma_{A'}\otimes\rho^{(i)}_{A^{(i)}_1E}\left(U^{(1,i)}\right)^\dagger,V^{(i)}\sigma_{A'}\otimes\proj 0_{\tilde A_1}\otimes\rho_E\left(V^{(i)}\right)^\dagger\Big)\nonumber\\
	 	&\le\sum_ip_iP\Big( U^{(1,i)}\sigma_{A'}\otimes\rho^{(i)}_{A^{(i)}_1E}\left(U^{(1,i)}\right)^\dagger,\sigma_{A'}\otimes\tau^{(i)}_{\tilde A_1}\otimes\rho_E\Big)\nonumber\\
	 	&\quad+\sum_ip_iP\Big(\sigma_{A'}\otimes\tau^{(i)}_{\tilde A_1}\otimes\rho_E,V^{(i)}\sigma_{A'}\otimes\proj 0_{\tilde A_1}\otimes\rho_E\left(V^{(i)}\right)^\dagger\Big)\nonumber\\
	 	&\le\sum_ip_iP\Big( U^{(1,i)}\sigma_{A'}\otimes\rho^{(i)}_{A^{(i)}_1E}\left(U^{(1,i)}\right)^\dagger,\sigma_{A'}\otimes\tau^{(i)}_{\tilde A_1}\otimes\rho_E\Big)+\varepsilon.
	 \end{align*}
	 The first inequality is the triangle inequality, the second one is Equation \eqref{eq:padembezz}. It remains to bound the first summand,
	 \begin{align*}
	 	\sum_ip_iP\Big( U^{(1,i)}\sigma_{A'}\otimes\rho^{(i)}_{A^{(i)}_1E}\left(U^{(1,i)}\right)^\dagger,\sigma_{A'}\otimes\tau^{(i)}_{\tilde A_1}\otimes\rho_E\Big)
	 	&=\sum_ip_iP\Big(\sigma_{A'}\otimes\rho^{(i)}_{A^{(i)}_1E},\sigma_{A'}\otimes\left(\left(U^{(1,i)}\right)^\dagger\tau^{(i)}_{\tilde A_1}U^{(1,i)}\right)\otimes\rho_E\Big)\nonumber\\
	 	&=\sum_ip_i P\Big(\sigma_{A'}\otimes\rho^{(i)}_{A^{(i)}_1E},\sigma_{A'}\otimes\tau_{\tilde A^{(i)}_1}\otimes\rho_E\Big)\nonumber\\
	 	&\le\sum_ip_iP\left(\rho^{(i)}_{A^{(i)}_1E},\tau_{\tilde A^{(i)}_1}\otimes\rho_E\right)\nonumber\\
	 	&\le\varepsilon,
	 \end{align*}
	 where the first inequality is the triangle inequality again, and the second one is Equation \eqref{eq:flat-decoup}. This shows that we achieved $2\varepsilon$-decoupling, i.e.
	 \begin{align}\label{eq:non-smooth-decoup}
	 P\left(\rho^f_{A_1E},\sigma_{A'}\otimes\proj 0_{\tilde A_1}\otimes\rho_E\right)\le2\varepsilon.
	 \end{align}
We also have to bound $\log|A_2|$, i.e. we need to make sure that
 	 \begin{align*}
 	 	\max_{i=1,...,Q}\left(H_0(A)_{\rho^{(i)}}-H_{\min}(A|E)_{\rho^{(i)}}\right)\le I_{\max}(E;A)_\rho+\mathcal{O}\left(\log\left(\frac 1 \varepsilon\right)\right).
 	 \end{align*}
 	 This is shown in \cite{berta2011quantum} in the last part of the proof of Theorem 3.10.
 	 Thereby the size of the remainder system is bounded by
 	 \begin{align*}
 	 	\log|A_2|=\frac{1}{2}I_{\max}(A;E)_{\hat\rho}+\log H_{0}(A)_{\hat{\rho}}+\mathcal{O}\left(\log\left(\frac 1 \varepsilon\right)\right).
 	 \end{align*}
 	 If we only want to use unitaries, we can complete all involved isometries to unitaries by adding an appropriate additional pure ancilla system.
\end{proof}

As an easy corollary we can derive a bound on the remainder system that involves the smooth max-mutual information in a way that is fit for deriving an the asymptotic expansion of Equation (15) in the main text.

\begin{customthm}{1'}[Catalytic decoupling]\label{prop:anc-decoup-bcr}
Let $\hat\rho_{AE}\in\St(\hi_A\otimes \hi_E)$ be a quantum state. Then, $\varepsilon$-catalytic decoupling can be achieved with remainder system size
	\begin{align*}
		\!\log|A_2|\!\le\frac{1}{2}I_{\max}^{\varepsilon-\delta}(A;E)_{\rho}+\log H_0(A)_{{\rho}}+\mathcal{O}\left(\log\left(\frac 1 \delta\right)\right).
	\end{align*}
	In addition, if we allow for the use of isometries instead of unitaries, the ancilla systems final state is $\varepsilon$ close to its initial state.
\end{customthm}
\begin{proof}
	Let $\hat{\rho}\in B_{\eta}(\rho)$ with $\eta=\varepsilon-\delta$ such that
	\begin{equation}
		I_{\max}^{\eta}(A:E)_\rho=I_{\max}(A:E)_{\hat{\rho}}.
	\end{equation}
	Define $\rho'=\Pi\hat{\rho}\Pi$, where $\Pi$ is the orthogonal projector onto the support of $\rho$. It follows from Uhlmann's theorem that $P(\rho,\rho')\le P(\rho,\hat{\rho})$. As the max-mutual information is non-increasing under projections (cf. \cite[Lemma B.19]{berta2011quantum}), it follows that
	\begin{equation}
			I_{\max}^{\eta}(A:E)_\rho=I_{\max}(A:E)_{\rho'}
		\end{equation}
	as well. Applying Supplemental Theorem \ref{thm:nonsmooth-anc-decoup-bcr} to $\rho'$ and an application of the triangle inequality yields the claimed bound.
\end{proof}

If we accept a slightly worse smoothing parameter for the leading order term, i.e. the max-mutual information, we can smooth the second term as well and replace the renyi-0 entropy by the max-entropy.

\begin{cor}
	Let $\rho_{AE}\in\St(\hi_A\otimes \hi_E)$ be a quantum state. Then, $\varepsilon$-catalytic decoupling can be achieved with remainder system size
		\begin{align*}
			\!\log|A_2|\!\le\frac{1}{2}I_{\max}^{\varepsilon'}(A;E)_{\rho}+\log H_{\max}^{{\varepsilon'}^2/2}(A)_{\rho}+\mathcal{O}(\log{\varepsilon'}\!),
		\end{align*}
		where $\varepsilon'=\varepsilon/6$.
\end{cor}

\begin{proof}
 	 To get the bound involving smooth entropy measures we will find a state $\hat\rho\in B_{2{\varepsilon'}}(\rho)$ such that $I_{\max}(E;A)_{\hat\rho}\le I_{\max}^{\varepsilon'}(E;A)_{\rho}$ and $H_0(A)_{\hat\rho}\le H^{{\varepsilon'}^2/2}_0(A)_{\rho}$. Let $\rho'_{AE} \in B_{\varepsilon'}(\rho_{AE})$ such that $I_{\max}(E;A)_{\rho'}=I_{\max}^{\varepsilon'}(E;A)_{\rho}$. Let $\Pi_A$ be a projection of minimal rank such that $H_0^{{\varepsilon'}^2/2}(A)_{\rho}\ge H_0(A)_{\rho''}$, with $\rho''=\Pi_A \rho_{AE}\Pi_A\in B_{\varepsilon'}(\rho_{AE})$. To see why such a projection exists, note that Supplemental Lemma \ref{lem:spectr} implies that there exists a state $\rho''_A$ such that
 	 	\begin{align}\label{eq:H0bound}
 	 		H_0(A)_{\rho''_A}=H_0^{{\varepsilon'},\tr}(A)_{\rho}\le H_0^{{\varepsilon'}^2\!/2}(A)_{\rho}
 	 	\end{align}
 	 	and $[\rho_A,\rho''_A]=0$, where the inequality is due to the equivalence lemma \ref{lem:equiv} of the trace distance and the purified distance. But for the case of commuting density matrices, i.e. the classical case, it is clear that the density matrix in a given trace distance neighborhood of $\rho$, that has minimal rank, is just equal to $\rho$ with the smallest eigenvalues set to zero. This implies that $\rho''_A$ can be chosen to have the form $\rho''_A=\Pi_A\rho_A\Pi_A$. It is easy to see that $P(\rho_{AE},\rho''_{AE})\le{\varepsilon'}$ where $\rho''_{AE}=\Pi_A\rho_{AE}\Pi_A$: Pick a purification $\ket{\rho''}_{AER}=\Pi_A\ket{\rho}_{AER}$ and observe that
 	 	\begin{align*}
 	 		F(\rho_{AE},\rho''_{AE})=\max_{\ket\sigma_{AER}}\left|\braket{\sigma}{\rho''}\right|=\max_{\ket\sigma_{AER}}\left|\bra{\sigma}\Pi\ket{\rho}\right|=\tr\Pi\rho=F(\rho_A,\rho''_A),
 	 	\end{align*}
 	 	where the fist equation is Uhlmann's theorem and the third equation follows from the saturation of the Cauchy-Schwarz inequality. We also use that $[\rho_A,\rho''_A]=0$ in the last equation. Now, we define $\hat\rho=\Pi_A\rho'\Pi_A$ and bound
 	 	\begin{align}\label{eq:annoying-dist-ineq}
 	 	P(\hat\rho_{AE},\rho_{AE})=P(\Pi_A\rho'_{AE}\Pi_A,\rho_{AE})= P( \rho'_{AE}, \Pi_A\rho_{AE}\Pi_A)=P( \rho'_{AE}, \rho''_{AE})\le P(\rho'_{AE}, \rho_{AE})+P(\rho_{AE}, \rho''_{AE})\le 2{\varepsilon'}.
 	 	\end{align}
 	 	The second equation follows easily by Uhlmann's theorem. According to \cite[Lemma B.19]{berta2011quantum} the max-mutual-information decreases under projections, i.e. we have
 	 	\begin{align*}
 	 		I_{\max}(E;A)_{\hat\rho}\le I_{\max}(E;A)_{\rho'}=I_{\max}^{\varepsilon'}(E;A)_{\rho}.
 	 	\end{align*}
 	 	Our choice of $\Pi_A$ gives
 	 	\begin{align*}
 	 		H^{{\varepsilon'}^2/2}_0(A)_{\rho}\ge H_0(A)_{\rho''_A}=\log\mathrm{rk}\Pi_A\ge H_0(A)_{\hat\rho},
 	 	\end{align*}
 	 	where the first inequality is Equation \eqref{eq:H0bound}.
 	 	Now, we apply Supplemental Theorem \ref{thm:nonsmooth-anc-decoup-bcr}, to $\hat\rho_{AE}$. Let $\rho^{(f)}_{A_1A_2E}$ be the final state when applying the resulting protocol to $\rho_{AE}$. Then, we get
 	 	\begin{align*}
		 	P\left(\rho^{(f)}_{A_1E},\sigma_{A'}\otimes\proj 0_{\tilde A_1}\otimes\rho_E\right)\le 6{\varepsilon'}
 	 	\end{align*}
 	 	by using Equations \eqref{eq:annoying-dist-ineq}, \eqref{eq:non-smooth-decoup} , the triangle inequality and the monotonicity of the purified distance under CPTP maps. 	 	
 \end{proof}


\subsubsection{Standard decoupling and comparison}

Let us look at an example of a state where the smooth min-entropy is almost zero and the smooth max-entropy is almost maximal to illustrate the significance of the randomization condition that is usually demanded for standard decoupling. To bound the max-entropy in the following example we need

\begin{slem}\label{lem:smoothing-flat-states}
    Let $0 < q < 1$ and $0 < \varepsilon^2 < 1-\sqrt{1-q}$, and let $\ket{\Phi}_{AB}$ be a maximally entangled state with $\dim A = \dim B = d$. Then, we have that
    \[ \hmin^{\varepsilon}(A|B)_{q\Phi} \leq -\log d + \log {\frac{1}{1-\varepsilon^2 - \sqrt{1-q}}}. \]
\end{slem}
\begin{proof}
    We can modify the SDP for the smooth min-entropy from~\cite[Proof of Lemma 5]{vitanov2012chain} to work with subnormalized states, by adding an extra dimension. The result is that given a state $\rho_{AB}$ with $\tr[\rho]=p$, the value of the following SDP is $2^{-\hmin^{\varepsilon}(A|B)_{\rho}}$:
\begin{center}
    \textbf{Primal problem:}
\end{center}
\[
    \begin{array}{ll}
        \mbox{minimize} & \tr[\sigma_B]\\[.2cm]
        \mbox{subject to} & \tilde{\rho}_{AB} \leq \mathds{1}_A \otimes \sigma_B\\
            &\tr\left[X \begin{pmatrix} \rho_{ABC} & \sqrt{1-p}\ket{\rho}\\ \sqrt{1-p} \bra{\rho} & 1-p \end{pmatrix} \right] \geq 1-\varepsilon^2\vspace{.1cm}\\
            &\tr[X] \leq 1\vspace{.1cm}\\
            & X = \begin{pmatrix} \tilde{\rho}_{ABC} & \ket{\psi}\\ \bra{\psi} & x \end{pmatrix}.
    \end{array}
\]
\vspace{.5cm}
\begin{center}
    \textbf{Dual problem:}
\end{center}
\[
    \begin{array}{ll}
        \mbox{maximize} & (1-\varepsilon^2)\mu - \lambda\\[.2cm]
        \mbox{subject to} & \mu \begin{pmatrix} \rho_{ABC} & \sqrt{1-p}\ket{\rho}\\ \sqrt{1-p} \bra{\rho} & 1-p \end{pmatrix} \leq \begin{pmatrix} E_{AB} \otimes \mathds{1}_C & 0 \\ 0 & 0 \end{pmatrix} + \begin{pmatrix} \lambda\mathds{1}_{ABC} & 0 \\ 0 & \lambda \end{pmatrix}\vspace{.1cm}\\
        & \tr_A[E_{AB}] \leq \mathds{1}_B.
    \end{array}
\]
\vspace{.5cm}

In the above, $\ket{\rho}_{ABC}$ is some fixed purification of $\rho_{AB}$. 

Now, to get the bound for $\rho = q \Phi$, we can choose $\mu = d$, $E_{AB} = d \Phi_{AB}$, and $\lambda = d\sqrt{1-q}$. The value of the dual problem for this choice of variables is then $d(1-\varepsilon^2) - d\sqrt{1-q}$. This is therefore a lower bound on $2^{-\hmin^\varepsilon(A|B)_{q\Phi}}$ and concludes the proof.
\end{proof}

\begin{ex}\label{ex:trivE}
	Define a probability distribution on $\{0,1,...,n\}$ by $p(0)=p_0$ and $p(i)=\frac{1-p_0}{n}$ for $i\neq 0$. Supplemental Lemma \ref{lem:equiv} shows that $H_{\min}^\varepsilon(p)\le H_{\min}^{\varepsilon(p),\tr}(p)$, where the superscript $\tr$ indicates that the non-smooth quantity is optimized over the trace distance ball instead of the purified distance ball. Considering that the min- and max-entropy are functions of the spectrum we can optimize over probability distributions only. The non-smooth min-entropy of $p$ is $H_{\min}(p)=-\log p_0$. Assume $p_0(1-\varepsilon+1/n)-1/n\ge 0$. Then, the best we can do for increasing this is obviously to reduce the probability of the outcome $0$. Take a sub-normalized probability distribution $q$ with $q(0)=q_0$ and $q(i)=p(i),\ i>0$. Then, we have $\delta(p,q)=p_0-q_0$, i.e. by Supplemental Lemma \ref{lem:uncond}
	\begin{align*}
	H_{\min}^{\varepsilon,\tr}(p)=-\log(p_0-\varepsilon),
	\end{align*}
	Assuming $\varepsilon\le p_0-(1-p_0)/n$.
	
	Using Supplemental Lemma \ref{lem:smoothing-flat-states} we get, assuming $\varepsilon^2\le 1-\sqrt{p_0}$, that
	
	\begin{align*}
			H_{\max}^{\varepsilon}(p)&\ge H_{\max}^{\varepsilon}((1-p_0)U(n))= -H_{\min}^\varepsilon(A|B)_{(1-p_0)\Phi}\ge \log n - \log {\frac{1}{1-\varepsilon^2 - \sqrt{p_0}}},
		\end{align*}
	where $U(n)$ denotes the uniform distribution on $n$ symbols.
	Putting in $p_0=1/2$ and $\varepsilon<1/15$ yields, after some calculations,
	\begin{align*}
		H_{\max}^{\varepsilon}(p)-H_{\min}^{\varepsilon}(p)&\ge H_{\max}^{\varepsilon}(p)-H_{\min}^{\varepsilon,\tr}(p)\ge\log(n)-\log\left(\frac{1}{(1-\sqrt{p_0}-\varepsilon^2)(p_0-\varepsilon)}\right)\\
		&\ge \log n-\log\left(\frac{10}{1-15\varepsilon}\right).
	\end{align*}
	\vspace{.1cm}
\end{ex}

The next theorem is a one-shot decoupling theorem for the partial trace with a bound on the remainder system involving smooth entropies. Plugging in the partial trace map into Theorem 3.1 in \cite{dupuis2014one} yields a priori the non-smooth $\log|A_2|\ge\frac 1 2\left(\log|A|-H_{\min}^{\varepsilon}(A|E)\right)$ for the remainder system when decoupling $A$ from $E$ in a state $\rho_{AE}$ despite the smoothness of the term depending on the map. This can be understood considering the fact that the Choi-Jamio\l kowski state of the partial trace is a tensor product of states with flat marginals, such that smoothing doesn't change much. For convenience we use~\cite[Theorem 3.1]{berta2011quantum} as a basic decoupling theorem.

\begin{sthm}
	Let $\rho_{AE}$ be a bipartite quantum state, and let $\hi_A\cong\hi_{A_1}\otimes\hi_{A_2}$ such that
	\begin{align*}
		\log|A_2|\ge \frac 1 2\left(H_{\max}^{\varepsilon}(A)_\rho-H_{\min}^\varepsilon(A|E)_\rho\right)-3\log\frac 1 \varepsilon.
	\end{align*}
	Then, we have
	\begin{align*}
		\intop_{\mathrm{U}(\hi_A)}P\left(\tr_{A_2}\left(U_A\rho_{AE}U_A^\dagger\right),\frac{1_{A_1}}{|A_1|}\otimes\rho_{E}\right)\D U_A\le \left(2+\sqrt{7}\right)\varepsilon\le 5 \varepsilon.
	\end{align*}
\end{sthm}

\begin{proof}
	Let $\tilde{\rho}_{AE}\in B_\varepsilon(\rho_{AE})$ such that $H_{\min}^\varepsilon(A|E)_\rho=H_{\min}(A|E)_{\tilde\rho}$, $\bar{\rho}_{A}\in B_\varepsilon(\rho_{A})$ such that $H_{0}^\varepsilon(A)_\rho=H_{0}(A)_{\bar\rho}$ and $\Pi_A$ the projection onto the support of $\bar\rho_A$.  Define the state $\hat{\rho}_{AE}=\Pi_A\tilde{\rho}_{AE}\Pi_A$. By Supplemental Lemma \ref{lem:lowrk} we can assume that $\bar{\rho}=\Pi_A\rho\Pi_A/\tr\rho$. $\rho$ is normalized, so a short calculation shows that
	\begin{align}\label{eq:trbound}
		F(\rho,\Pi_A\rho\Pi_A/\tr(\Pi_A\rho))=\sqrt{\tr(\Pi_A\rho)}\quad\Rightarrow\quad\tr(\Pi_A\rho)\ge 1-\varepsilon^2
	\end{align}
	via the definitions of $H_0^\varepsilon$ and the purified distance.
	 By the triangle inequality we get
	\begin{align}\label{eq:bound1}
		P(\rho,\hat{\rho})\le P(\rho,\bar{\rho})+P(\bar{\rho},\hat{\rho})\le\varepsilon+P(\Pi\tilde\rho\Pi, \Pi\rho\Pi/\tr(\Pi\rho)).
	\end{align}
	We continue to bound the last term. We have
	\begin{align*}
		F(\Pi\tilde\rho\Pi, \Pi\rho\Pi/\tr(\Pi\rho))&=\frac{1}{\sqrt{tr(\Pi\rho)}}\left\|\sqrt{\Pi\tilde\rho\Pi}\sqrt{\Pi\rho\Pi}\right\|_1\nonumber\\
		&=\frac{1}{\sqrt{tr(\Pi\rho)}}\left(F(\Pi\tilde\rho\Pi, \Pi\rho\Pi)-\sqrt{(1-\tr(\Pi\rho))(1-\tr(\Pi\tilde\rho))}\right)\nonumber\\
		&\ge\left(F(\tilde\rho, \rho)-\sqrt{(1-\tr(\Pi\rho))(1-\tr(\Pi\tilde\rho))}\right).
	\end{align*}
	The last step, i.e. that the generalized fidelity does not decrease under projections, follows easily from the fact that the regular fidelity does not decrease under CPTP maps. To bound the remaining term, note that
	\begin{align*}
		\sqrt{(1-\tr(\Pi\rho))(1-\tr(\Pi\tilde\rho))}+\sqrt{\tr(\Pi\rho)\tr(\Pi\tilde\rho)}\ge F(\rho,\tilde\rho)\ge \sqrt{1-\varepsilon^2}
	\end{align*}
	by the monotonicity of the fidelity under CPTP maps. Let $\phi, \theta\in[0,\pi/2]$ such that $\cos^2\phi=\tr(\Pi\rho)$ and $\cos^2\theta=\tr(\Pi\tilde{\rho})$. Then, some trigonometric identities yield $\sin(\phi-\theta)\le\varepsilon$, i.e. in particular $\phi\le\theta+\arcsin(\varepsilon)$. Using this bound, Equation \eqref{eq:trbound} and some more trigonometry yields
	\begin{align*}
		\sqrt{(1-\tr(\Pi\rho))(1-\tr(\Pi\tilde\rho)}=\sin\phi\sin\theta\le3\varepsilon^2\sqrt{1-\varepsilon^2}.
	\end{align*}
	This implies now that
	\begin{align*}
		F(\Pi\tilde\rho\Pi, \Pi\rho\Pi/\tr(\Pi\rho))\ge\sqrt{1-\varepsilon^2}(1-3\varepsilon^2)\quad\Rightarrow\quad P(\Pi\tilde\rho\Pi, \Pi\rho\Pi/\tr(\Pi\rho))\le \sqrt{7\varepsilon^2-15\varepsilon^4+9\varepsilon^6}\le\sqrt{7}\varepsilon.
	\end{align*}
	Together with Equation \eqref{eq:bound1} this yields $P(\rho,\hat{\rho})\le(1+\sqrt{7})\varepsilon$. Considering $\hat{\rho}\in\St(\supp\bar\rho_A\otimes\hi_E)$, an application of~\cite[Theorem 3.1]{berta2011quantum} together with Supplemental Lemma \ref{lem:equiv} results in the following. If $A=A_1A_2$ and 
	\begin{align}\label{eq:entrcond}
		\log|A_2|\ge \frac{1}{2}\left(\log\mathrm{rk}(\bar\rho_A)-H_{\min}(A|E)_{\hat{\rho}}\right)-\log\frac 1 \varepsilon,
	\end{align}
	then we have
	\begin{align*}
		\intop_{\mathrm{U}(\hi_A)}P\left(\tr_{A_2}\left(U_A\hat\rho_{AE}U_A^\dagger\right),\frac{1_{A_1}}{|A_1|}\otimes\hat\rho_{E}\right)\D U_A\le \varepsilon.
	\end{align*}
	The last equation implies, together with the triangle inequality, that 
		\begin{align*}
				\intop_{\mathrm{U}(\hi_A)}P\left(\tr_{A_2}\left(U_A\rho_{AE}U_A^\dagger\right),\frac{1_{A_1}}{|A_1|}\otimes\hat\rho_{E}\right)\D U_A\le \left(2+\sqrt{7}\right)\varepsilon.
			\end{align*}
	Equation \eqref{eq:entrcond} together with Supplemental Lemma \ref{lem:hminproj} and \ref{lem:H0Hmax} implies the claimed bound on the remainder system size.
\end{proof}

In the following we present a correction of the converse for decoupling by CPTP map, Corollary 4.2, from \cite{dupuis2014one}, a slightly tighter version of Proposition 4.

In the context of decoupling by partial trace we observed that it makes a big difference whether we demand that the decoupled system is randomized as well, i.e. that it is left in the maximally mixed state. This stops making sense in the context of decoupling by a general CPTP map $\mathcal{T}$, as the maximally mixed state might not even be in the \emph{range} of $\mathcal{T}$. Instead one can demand randomizing in the sense that is achieved in the direct result in \cite{dupuis2014one}, i.e. $\mathcal T_{A\to B}\left(\rho_{AE}\right)\approx\mathcal T_{A\to B}\left(\frac{1_A}{|A|}\right)\otimes\rho_E$.

The following theorem from \cite{dupuis2014one} is already a converse statement for decoupling by CPTP map.

\begin{sthm}\cite[Theorem 4.1]{dupuis2014one}\label{thm:cptp-decoup-con-omega}
	 Let $\rho\in\mathcal{S}(\hi_A\otimes\hi_E)$ and $\mathcal{T}: \End{\hi_A}\to\End{\hi_B}$ a CPTP map such that
	 \begin{align*}
	 	\left\|\mathcal{T}\otimes\mathrm{id}_E(\rho_{AE})-\mathcal{T}(\rho_A)\otimes\rho_E\right\|_1\le \varepsilon.
	 \end{align*}
	 Then, we have
	 \begin{align*}
	 	H_{\min}^{2\sqrt{6\varepsilon''+2\varepsilon}+2\sqrt{\varepsilon'}+\varepsilon''}(A|E)_{\rho}+H_{\max}^{\varepsilon''}(A|B)_{\omega}\ge \log\varepsilon'
	 \end{align*}
	 for all $\varepsilon', \varepsilon''>0$, where $\omega_{AB}=\mathrm{id}_A\otimes\mathcal{T}_{A'\to B}(\rho_{AA'})$ with $\hi_{A'}\cong\hi_A$ and $\rho_{AA'}$ a purification of $\rho_A$.
\end{sthm}

It involves, however, the term $H_{\max}^{\varepsilon''}(A|B)_{\omega}$ that depends on both the state and the CPTP-map. Unfortunately the proof of the Corollary following this theorem, Corollary 4.2, contains a mistake and the statement is incorrect as it is stated in \cite{dupuis2014one}. The reason for this is that the converse, Corollary 4.2, does not assume decoupling \emph{and} randomizing, while the direct result, \cite[Theorem 3.1]{dupuis2014one}, provides a condition for exactly that. Adding this condition to the statement of Corollary 4.2 renders it true and we give a proof of it in the following.

\begin{customprop}{4'}[Corrected version of Corollary 4.2 in~\cite{dupuis2014one}]\label{thm:cptp-decoup-con}
		Let $\rho_{AE}\in\mathcal{S}(\hi_A\otimes\hi_E)$ and let the CPTP map $\mathcal{T}: \End{\hi_A}\to\End{\hi_B}$ be such that
	 	\begin{align*}
	 		\intop_{\mathcal{U}(\hi_A)}P\left(\mathcal T_{A\to B}(U_A\rho_{AE}U_A^\dagger),\tau_B\otimes\rho_E\right)\D U_A\le \varepsilon.
	 	\end{align*}
Then, we have
		 \begin{align*}
		 H_{\min}^{4\sqrt{6\varepsilon''+2\varepsilon}+2\varepsilon''+\varepsilon'''}(A|E)_\rho+H_{\max}^{\varepsilon''}(A|B)_{\tau}\ge-10\log\left(\frac{1}{\varepsilon'''}\right)-7
		 \end{align*}
		 for all $\varepsilon''', \varepsilon''>0$, where $\tau=\mathcal{T}_{A'\to B}(\phi^+_{AA'})$ is the Choi-Jamio\l kowski state of $\mathcal{T}$.
\end{customprop}

\begin{proof}
	For $\delta\ge 0$ arbitrary, let $\mathcal{D}\subset \mathcal{U}_A$, $|\mathcal{D}|<\infty$ be a $\delta$-net in $\mathcal{U}_A$ in the operator norm, i.e. a finite subset such that for all $U\in \mathcal{U}_A$ there exists $V\in\mathcal{D}$ such that $\|U-V\|_\infty\le\delta$. Now, define the state
	\begin{align*}
		\tilde{\rho}_{AEU}=\sum_{U_A\in\mathcal{D}}q_{{}_{U_A}}U_A\rho_{AE}U_{A}^\dagger\otimes\proj{U_A}_U
	\end{align*}
	where $\hi_U=\C^{|\mathcal D|}$,
	\begin{align*}
		q_{{}_{U_A}}=\mu\left(\left\{U\in\mathcal{U}_A\Big| \|U-U_A\|_\infty\le\|U-V\|_\infty\,\forall V\in\mathcal{D}\right\}\right),
	\end{align*}
	and $\mu$ denotes the Haar measure.
	The assumption implies that $\mathcal T_{A\to B}$ decouples $A$ from $EU$. To see this, note that
	\begin{align} \label{eq:deltanet}
		P\left(\mathcal T_{A\to B}\tilde{\rho}_{AEU},\mathcal T_{A\to B}\tilde{\rho}_{A}\otimes\tilde{\rho}_{EU}\right)
		&=\frac{1}{|\mathcal{D}|}\sum_{U_A\in\mathcal{D}}P\left(\mathcal{T}_{A\to B}^{U_A}\rho_{AE},\mathcal{T}_{A\to B}\tilde{\rho}_A\otimes\rho_E\right)\nonumber\\
		&=\intop_{\mathcal{U}(\hi_A)}P\left(\mathcal{T}_{A\to B}^{D(U_A)}\rho_{AE},\mathcal{T}_{A\to B}\tilde{\rho}_A\otimes\rho_E\right)\D U\nonumber\\
		&\le\intop_{\mathcal{U}(\hi_A)}P\left(\mathcal{T}_{A\to B}^{U_A}\rho_{AE},\mathcal{T}_{A\to B}\tilde{\rho}_A\otimes\rho_E\right)\D U+\sqrt{2\delta},
	\end{align}
	where $\mathcal{T}_{A\to B}^{U_A}: X\mapsto \mathcal{T}_{A\to B}(U_AXU_A^\dagger)$ and $D(U_A)=\mathrm{minarg}_{V\in\mathcal{D}}\|U_A-V\|_\infty$. The last inequality follows easily using the $\delta$-net-property, the triangle inequality, the fact that the purified distance decreases under CPTP maps and Supplemental Lemma \ref{lem:smallU}. By assumption we then have
	\begin{align*}
			P\left(\mathcal T_{A\to B}\tilde{\rho}_{AEU},\mathcal T_{A\to B}\tau_A\otimes\tilde{\rho}_{EU}\right)\le\varepsilon+2\sqrt{2\delta},
		\end{align*}
	as $P(\mathcal T_{A\to B}\tau_A, \mathcal{T}_{A\to B}\tilde{\rho}_A)\le\sqrt{2\delta}$ by a similar argument as in Equation \ref{eq:deltanet}.
	Using Supplemental Lemma \ref{lem:equiv} and applying Supplemental Theorem \ref{thm:cptp-decoup-con-omega} to this situation, i.e. the map $\mathcal{T}_{A\to B}$ that decouples system $A$ from $EU$ applied to the state $\tilde{\rho}$, we get
	\begin{align}\label{eq:from-convlem}
		H_{\min}^{2\sqrt{6\varepsilon''+2\varepsilon+8\delta}+2\sqrt{\varepsilon'}+\varepsilon''}(A|EU)_{\tilde{\rho}}+H_{\max}^{\varepsilon''}(A|B)_{\tau}\ge-\log\left(\frac{1}{\varepsilon'}\right),
	\end{align}
	whit the Choi-Jamio\l kowski state $\tau_{AB}=\mathcal{T}_{A'\to B}\phi^+_{AA'}$. Let $\eta=2\sqrt{6\varepsilon''+4\varepsilon}+2\sqrt{\varepsilon'}+\varepsilon''$. The min-entropy term can be transformed using the chain rules for smooth entropies \cite{vitanov2012chain},
	\begin{align}\label{eq:chainrules}
	H_{\min}^{\eta}(A|EU)_{\tilde{\rho}}
	&\le H^{2\eta+\varepsilon^{(3)}}_{\min}(AU|E)_{\tilde{\rho}}-H_{\min}(U|E)_{\tilde{\rho}}+\log\left(\frac{2}{(\varepsilon^{(3)})^2}\right)\nonumber\\
	&= H^{2\eta+\varepsilon^{(3)}}_{\min}(AU|E)_{\rho\otimes\tau_U}-H_{\min}(U)_{\tau_U}+\log\left(\frac{2}{(\varepsilon^{(3)})^2}\right)\nonumber\\
	&\le H^{2\eta+\varepsilon^{(3)}+2\varepsilon^{(4)}}_{\min}(A|E)_{\rho}+H_{\max}(U|AE)_{\rho\otimes\tau_U}-H_{\min}(U)_{\tau_U}+\log\left(\frac{2}{(\varepsilon^{(3)})^2}\right)+3\log\left(\frac{2}{(\varepsilon^{(4)})^2}\right)\nonumber\\
	&=H^{2\eta+\varepsilon^{(3)}+2\varepsilon^{(4)}}_{\min}(A|E)_{\rho}+H_{\max}(U)_{\tau_U}-H_{\min}(U)_{\tau_U}+\log\left(\frac{2}{(\varepsilon^{(3)})^2}\right)+3\log\left(\frac{2}{(\varepsilon^{(4)})^2}\right)\nonumber\\
	&\le H^{2\eta+\varepsilon^{(5)}}_{\min}(A|E)_{\rho}+8\log\left(\frac{1}{\varepsilon^{(5)}}\right)+13,
	\end{align}
	where we used a chain rule in the first inequality, in the second line that $U$ is independent from $E$, the invariance of the smooth entropies under isometries and that there exists a controlled unitary $V_{UA}$ such that $V_{UA}\tilde\rho_{AEU}V_{UA}^\dagger=\rho_{AE}\otimes\tau_U$, and another chain rule in the  fourth line. In the last line we set $\varepsilon^{(5)}=2\varepsilon^{(4)}+\varepsilon^{(3)}$ and $\varepsilon^{(4)}=2\varepsilon^{(3)}/3$ to get an optimal error term. Combining Equations \eqref{eq:from-convlem} and \eqref{eq:chainrules} we get
	\begin{align*}
		H_{\min}^{4\sqrt{6\varepsilon''+2\varepsilon+4\delta}+4\sqrt{\varepsilon'}+2\varepsilon''+\varepsilon^{(5)}}(A|E)_\rho+H_{\max}^{\varepsilon''}(A|B)_{\tau}\ge-\log\left(\frac{1}{\varepsilon'}\right)-8\log\left(\frac{1}{\varepsilon^{(5)}}\right)-13.
	\end{align*}
	Fixing $\varepsilon'''=4\sqrt{\varepsilon'}+\varepsilon^{(5)}$ and optimizing the logarithmic error term yields
	\begin{align*}
	H_{\min}^{4\sqrt{6\varepsilon''+2\varepsilon+4\delta}+2\varepsilon''+\varepsilon'''}(A|E)_\rho+H_{\max}^{\varepsilon''}(A|B)_{\tau}\ge-10\log\left(\frac{1}{\varepsilon'''}\right)-7.
	\end{align*}
	As $\delta$ was arbitrary, we can take the limit\footnote{The limit $\delta\to 0$ exists, as the min-entropy term that depends on $\delta$ is nondecreasing in $\delta$ and bounded from below.} $\delta\to 0$, which concludes the proof.
\end{proof}

Theorem 3 follows as an easy corollary.

\begin{customprop}{4}
	Let $\rho_{AE}\in\mathcal{S}(\hi_A\otimes\hi_E)$ and let the CPTP map $\mathcal{T}: \End{\hi_A}\to\End{\hi_B}$ be such that
	 	\begin{align*}
	 		\intop_{\mathcal{U}(\hi_A)}P\left(\mathcal T_{A\to B}(U_A\rho_{AE}U_A^\dagger),\mathcal T_{A\to B}(\tau_{A})\otimes\rho_E\right)\D U_A\le \varepsilon.
	 	\end{align*}
Then, we have
\begin{align*}
H_{\min}^{15\sqrt{\varepsilon}}(A|E)_\rho+H_{\max}^{\varepsilon}(A|B)_{\tau}\ge -10\log\left(\frac{1}{\varepsilon}\right)-7,
\end{align*}
where $\tau=\mathcal{T}_{A'\to B}(\phi^+_{AA'})$ is the Choi-Jamio\l kowski state of $\mathcal{T}$.
\end{customprop}

\begin{proof}
Setting $\varepsilon=\varepsilon'=\varepsilon''=\varepsilon'''$ in Theorem \ref{thm:cptp-decoup-con} and bounding $\varepsilon\le\sqrt{\varepsilon}$ yields the result.
\end{proof}


\subsubsection{Asymptotic expansion}

Here we give the necessary definitions and point to the relevant references to derive the asymptotic expansion given in Equation (15) in the main paper.

\begin{sdefn}[Quantum relative entropy \cite{umegaki1962conditional} and quantum information variance \cite{tomamichel2013hierarchy}]
	For quantum states $\rho, \sigma\in\St(\hi)$ the quantum relative entropy is defined as follows:
	\begin{align*}
		D(\rho\|\sigma)=\begin{cases}
		\tr\rho\left(\log\rho-\log\sigma\right)&\ \supp\rho\subset\supp \sigma\\
		\infty& \mathrm{ else}
		\end{cases},
	\end{align*}
	i.e. as the expectation of $\log\rho-\log\sigma$ with respect to $\rho$.
	The quantum information variance is the corresponding variance,
	\begin{align*}
		V(\rho\|\sigma)=\begin{cases}
					\tr\rho\left(\log\rho-\log\sigma\right)^2&\ \supp\rho\subset\supp \sigma\\
					\infty& \mathrm{ else}
				\end{cases}.
	\end{align*}
\end{sdefn}

The von Neumann entropy and derived quantities can be expressed in terms of the quantum relative entropy and thereby given a corresponding variance. In particular we have that
\begin{align*}
I(A;B)_\rho=D(\rho_{AB}\|\rho_A\otimes\rho_B).
\end{align*}
Consequently we define $V(A;B)_\rho=V(\rho\|\rho_A\otimes\rho_B)$.

Equation (15) in the main paper makes use of the cumulative normal distribution,
\begin{align*}
\Phi(x)=\frac{1}{\sqrt{2\pi}}\int_{-\infty}^{x}\exp(-y^2/2)\D y.
\end{align*}
Note that this function is invertible.

The derivation of the asymptotic expansion is not detailed here, as it is completely analogous to the derivation in \cite[Section VI]{tomamichel2013hierarchy}.


\subsubsection{Quantum state redistribution from catalytic decoupling}

In this section we show how to apply any decoupling with ancilla protocol to quantum state redistribution. Let us first define the task of quantum state redistribution (QSR).

\begin{sdefn}[Quantum state redistribution \cite{yard2009optimal,horodecki2007quantum}]
\textcolor{white}{placeholder}
\begin{itemize}
\item Let $\ket{\psi}_{ABCR}$ be a four party quantum state where Alice holds systems $A$ and $C$, Bob holds System $B$ and a referee holds system $R$. An \emph{$\varepsilon$-quantum state redistribution protocol with communication cost $q$} is a protocol in which Alice performs some encoding operation on here shares of $\ket{\psi}_{ABCR}$ and some resource state $\ket{\phi}_{A'B'}$ shared between Alice and Bob, then she sends a quantum register $C_2$ of size $\log|C_2|=q$ to Bob who performs some decoding operation such that the final state is $\ket{\tilde{\psi}}_{ABCR}\otimes\ket{\tilde\phi}_{A''B''}$ with $P(\ket{\psi}_{ABCR},\ket{\tilde\psi}_{ABCR})\le\varepsilon$ and Alice holds $A$, Bob holds $B$ and $C$ and the referee still holds $R$.
\item For trivial system $A$, i.e. $\hi_A=\C$, the task is called \emph{quantum state merging}, for $\hi_B=\C$ \emph{quantum state splitting}.
\end{itemize}
\end{sdefn}

Asymptotically QSR can be achieved with a quantum communication cost of $I(R;C|A)_\psi$ \cite{yard2009optimal}, i.e. there exists a sequence of QSR protocols for $\ket{\psi}_{ABCR}^{\otimes n}$ with quantum communication cost $q_n$ such that
\begin{align*}
\lim_{n\to\infty}\frac{1}{n}q_n=I(R;C|A)_\psi.
\end{align*}

Anshu \etal~\cite{anshu2014near} define the following quantity that that characterizes the quantum communication cost of one-shot QSR:
\begin{align*}
I_{\max}^\varepsilon(R;C|A)_{\psi}=\inf I_{\max}(RA;CA')_{U_{ACA'}\rho_{RACA'} U^\dagger_{ACA'}},
\end{align*}
where the infimum is taken over ancilla systems $A'$, states $\sigma_{A'}$, states $\rho\in B_\varepsilon(\psi_{RAC}\otimes\sigma_{A'})$ and unitaries $U_{ACA'}\in\mathcal{U}(\hi_{ACA'})$ such that $\tr_{CA'}U_{ACA'}\rho_{RACA'} U^\dagger_{RACA'}\in B_\varepsilon(\psi_{RA})$. We call this quantity the smooth conditional max-mutual-information. Note that in \cite{anshu2014near} it is denoted by $Q^\varepsilon_\psi$. Using the same minimization idea we can get a QSR protocol, that improves over the naive use of a state splitting or state merging protocol to achieve QSR, from any decoupling theorem. The special case of state merging is presented as Proposition 3 in the main text.

\begin{sthm}[Quantum state redistribution from decoupling]\label{thm:decouptoqsr}
Quantum state redistribution for a state $\ket{\psi}_{ABCR}$ can be achieved up to a purified distance error of $3\varepsilon$ with a quantum communication cost of
\begin{align}\label{eq:q-quantity}
q^\varepsilon(R;C|A)_\psi=\inf R_c^\varepsilon(AR;CA'')_{U_{ACA''}\rho_{ARCA''}U_{ACA''}^\dagger},
\end{align}
where the infimum is taken over ancilla systems $A''$, states $\sigma_{A''}$, states $\rho_{ARCA''}\in B_\varepsilon(\psi_{RAC}\otimes\sigma_{A''})$ and unitaries $U=U_{ACA''}\in\mathcal{U}(\hi_{ACA''})$ such that $\tr_{CA''}U_{ACA''}\rho_{RACA''} U^\dagger_{ACA''}\in B_\varepsilon(\psi_{RA})$. For state merging the quantity $q^\varepsilon(R;C|A)_\psi$ reduces to $R_c^\varepsilon(R;C)$.
\end{sthm}

The problem that the infimum is taken over unbounded Hilbert space dimensions, and therefore might not be achievable using a finite-dimensional Hilbert space, is artificial in view of the fact that any one-shot protocol has a communication cost $q\in\log\N$, which is discrete, i.e. the infimum is actually a minimum.

The proof is an adaptation of the protocol used in \cite{anshu2014near}, Theorem 4.2, run backwards.
\begin{proof}
Let $\hi_{A''}$, $\sigma_{A''}\in\St(\hi_{A''})$, $\rho\in B_\varepsilon(\psi_{RAC}\otimes\sigma_{A''})$ and $U_{AC{A''}}\in\mathcal{U}(\hi_{AC{A''}})$ be a tuple that saturates the infimum in Equation \eqref{eq:q-quantity}. If such does not exist because the infimum is taken over unbounded finite Hilbert space dimensions, take a tuple that saturates the infimum up to $\varepsilon$. Now, consider the following protocol:
\begin{enumerate}
		\item Starting point of the protocol is that Alice, Bob and the Referee share a state $\psi_{ABCD}\otimes\sigma_{A''B''}\otimes\tilde{\rho}_{A'B'}$, where $\sigma_{A''B''} $ is a purification of $\sigma_{A''}$, $\tilde\rho_{A'B'}$ is a purification of any state $\tilde\rho_{A'}$ that Alice will need for decoupling, and Alice holds systems $AC{A''}{A'}$, Bob holds systems $B{B''}{B'}$ and the Referee holds $R$.
		\item Alice applies the unitary $U_{AC{A''}}$
		\item Alice takes the $\varepsilon$-decoupling isometry $V_{C{A''}{A'}\to C_1C_2}$ that was constructed for decoupling systems $C{A''}A'$ of the state $U_{AC{A''}}\rho_{RAC{A''}} U^\dagger_{AC{A''}}$ from $AR$ and runs it on her state $U_{AC{A''}}\left(\psi_{RAC}\otimes \sigma_{A''}\otimes\tilde{\rho}_{A'}\right) U^\dagger_{AC{A''}}.$ She then sends the $q^\varepsilon(R;C|A)_\psi$-qbit remainder system $C_2$ to Bob. The decoupling isometry with these properties exists by assumption.
		\item For Alice's and the Referee's joint state 
		\begin{align*}
		\xi_{C_1AR}=\tr_{C_2}VU\left(\psi_{RAC}\otimes \sigma_{A''}\otimes\tilde\rho_{A'}\right) U^\dagger V^\dagger
		\end{align*}
		the triangle inequality for the purified distance yields $P(\xi_{C_1AR},\xi_{C_1}\otimes\psi_{AR})\le3\varepsilon$. So according to Uhlmann's theorem Bob can apply an isometry such that the final state of the protocol is $3\varepsilon$-close to $\psi_{ABCR}\otimes \xi'_{C_1C_1'}$ in purified distance, where $\xi'_{C_1C_1'}$ is a purification of $\xi_{C_1}$.
\end{enumerate}
For state merging, i.e. the case of trivial $A$, the ancilla $A''$ becomes unnecessary and the unitary $U_{AC{A''}}$ can be taken to be equal to the identity. This yields the claimed improvement.
\end{proof}

Together with Theorem \ref{thm:anc-decoup-supp} this recovers the result from \cite{anshu2014near} that one-shot quantum state redistribution is achievable with a communication cost of $I_{\max}^\varepsilon(R;C|A)_{\psi}$ plus lower order terms.


\end{document}